\documentclass{article}
\date{}

\usepackage{bbm}
\usepackage{graphicx,tikz, color}
\usepackage{amsmath,amsfonts,amssymb,amsthm,latexsym}
\usepackage{enumerate}
\usepackage{epsfig}
\usepackage[utf8]{inputenc}

\theoremstyle{definition}
\newtheorem{theorem}{Theorem}[section]
\newtheorem{definition}[theorem]{Definition}

\newtheorem{corollary}[theorem]{Corollary}
\newtheorem{remark}[theorem]{Remark}
\newtheorem{lemma}[theorem]{Lemma}
\newtheorem{claim}[theorem]{Claim}

\newtheorem{example}[theorem]{Example}
\newtheorem{proposition}[theorem]{Proposition}

\newenvironment{claimproof}[1]{{\it\noindent{Proof.}}\space#1}{\footnotesize \hfill \ensuremath{(\square)} \medskip}

\def\Zmap{\mu}
\let\tm=\times
\let\sse=\subseteq
\let\ov=\overline
\def\zd{,\ldots,}
\def\tms{\tm\dots\tm}
\def\vc#1#2{#1_1\zd #1_{#2}}
\let\sg=\sigma

\def\bA{{\textbf A}}
\def\bB{{\textbf B}}
\def\bC{{\textbf C}}
\def\cB{\mathcal{B}}

\def\bF{\textbf{F}}
\def\bG{\textbf{G}}

\def\cH{\mathcal{H}}
\def\cM{\mathcal{M}}
\def\cT{\mathcal{T}}
\def\cS{\mathcal{S}}
\def\zA{\mathbb{A}}
\def\zB{\mathbb{B}}
\def\zZ{\mathbb{Z}}
\def\zR{\mathbb{R}}

\def\zG{\mathbb{G}}
\def\rel{R}
\def\relo{Q}

\def\GGB{G(\cH^\bB)}
\def\GA{\cH^\bA}
\def\GB{\cH^\bB}

\DeclareMathOperator{\CSP}{CSP}
\DeclareMathOperator{\PCSP}{PCSP}
\DeclareMathOperator{\Pol}{Pol}
\DeclareMathOperator{\Hom}{Hom}
\DeclareMathOperator{\img}{img}
\DeclareMathOperator{\sel}{sel}

\title{Discrete Homotopy and Promise Constraint Satisfaction Problem}
\author{Arash Beikmohammadi and Andrei A. Bulatov\thanks{The second author is supported by an NSERC Discovery grant}}{ }{}{}{} 

\begin{document}

\maketitle

\begin{abstract}
The Promise Constraint Satisfaction Problem (PCSP for short) is a generalization of the well-studied Constraint Satisfaction Problem (CSP). The PCSP has its roots in such classic problems as the Approximate Graph Coloring and the $(1+\varepsilon)$-Satisfiability problems. The area received much attention recently with multiple approaches developed to design efficient algorithms for restricted versions of the PCSP, and to prove its hardness.

One such approach uses methods from Algebraic Topology to relate the complexity of the PCSP to the structure of the fundamental group of certain topological spaces. In this paper, we attempt to develop a discrete analog of this approach by replacing topological structures with combinatorial constructions and some basic group-theoretic concepts. We consider the `one-dimensional' case of the approach. We introduce and prove the basics of the framework, show how it is related to the complexity of the PCSP, and obtain several hardness results, including the existing ones, as applications of our approach. The main hope, however, is that this discrete variant of the topological approach can be generalized to the `multi-dimensional' case needed for further progress.
\end{abstract}


\section{Introduction}\label{sec:intro}

The Constraint Satisfaction Problem (CSP for short), where the goal is to decide the existence of a homomorphism between relational structures, has been a very successful research area over the last several decades. Developments in the study of the CSP prompted attention to generalizations and variations of the problem. One of such important generalizations is the Promise CSP (PCSP) introduced in \cite{Aus2017}, in which for three relational structures $\bA,\bB,\bC$ such that $\bA$ is homomorphic to $\bB$, the goal is to decide whether $\bC$ is homomorphic to $\bA$ or is \emph{not} homomorphic to $\bB$. The PCSP has received much attention in recent years.

One of the main special cases that motivated the framework of the PCSP is the \emph{Approximate Graph Coloring Problem} (AGCP): Given a graph $G$, decide whether a graph is $k$-colorable or not even $\ell$-colorable for natural numbers $k, \ell$ with $k \leq \ell$. Overall, the belief is that like the $k$-Coloring problem, the AGCP is hard for almost all combinations of $k$ and $\ell$. The AGCP has been intensively studied even before the introduction of the PCSP, see, e.g., \cite{Gar1976, Kha2000, Gur2004, Hua2013,Bra2016,Din2006, Gur2020}, although only partial results had been obtained. An instance of the AGCP can be viewed as an instance $(K_k,K_\ell,G)$ of the PCSP, and therefore can also be tackled by the novel approaches developed for the PCSP. This resulted in much improved conditions on $k,\ell$ that guarantee hardness \cite{Bul2019,Wro2020}. The AGCP can be made more general by replacing $K_k$ and $K_\ell$ in its formulation with arbitrary graphs $G$ and $H$ (provided $G$ is homomorphic to $H$). It was recently proved in \cite{Avvakumov25} that the problem is NP-hard for any pair of non-bipartite 4-colorable graphs $G,H$. The research is ongoing.

When the target structures $\bA,\bB$ are fixed, the restricted version of the PCSP is denoted $\PCSP(\bA,\bB)$. Following the success of the program to classify the complexity of restricted versions of the CSP \cite{Bul2017,Zhu2020}, the principal goal of the research in PCSP is to find a complexity classification for problems of the form $\PCSP(\bA,\bB)$.
While the PCSP research borrows many ideas from the CSP, it appears that it requires a much more diverse variety of methods. This includes algorithms. While the majority of CSP algorithms are purely algebraic, the PCSP often requires LP- and SDP-based relaxations \cite{Bra2018,Bra2021,Ciardo23:CLAP,Ciardo23:semidefinite,Brakensiek23:SDP}. Proving hardness of PCSPs has also proved to be hard. This is witnessed by the AGCP that in spite of a range of tools thrown at it remains widely open~\cite{ciardo2023approximate,Ciardo23:semidefinite}. The PCSP also has a number of interesting structural features, see, e.g., \cite{Barto19:promises}. For a recent survey of the field, see \cite{Kro2022}.

In this paper we attempt to advance the methods of proving the hardness of the PCSP.

\smallskip
\noindent\textbf{Algebraic approach.}
One of the tools widely used to study the CSP and also applied for the PCSP is the algebraic approach, \cite{Bul2005,Barto21:algebraic}. This approach is based on the concept of a \emph{polymorphism} of a pair of relational structures $(\bA,\bB)$ such that $\bA$ is homomorphic to $\bB$ (aka a \emph{PCSP-template}). A polymorphism of $(\bA,\bB)$ is a homomorphism $f:\bA^n\to\bB$, where $\bA^n$ denotes the $n$-th direct power of $\bA$. It is also often convenient to think of $f$ as a function $f(\vc xn)$ with $n$ input variables taking values from $\bA$. The set of all polymorphisms of $(\bA,\bB)$ is denoted $\Pol(\bA,\bB)$ and is a \emph{minion}, that is, a set of functions closed under renaming of variables. For two PCSP-templates $(\bA_1,\bB_1),(\bA_2,\bB_2)$ a mapping from $\Pol(\bA_1,\bB_1)$ to $\Pol(\bA_2,\bB_2)$ that preserves the arity of the polymorphisms and renaming the variables is called a \emph{minion homomorphism}. Minion homomorphisms are very important in the study of the complexity of the PCSP, because, as is proved in \cite{Barto21:algebraic}, if there exists a minion homomorphism from $\Pol(\bA_1,\bB_1)$ to $\Pol(\bA_2,\bB_2)$ then $\PCSP(\bA_2,\bB_2)$ is polynomial time reducible to $\PCSP(\bA_1,\bB_1)$.

One of the very few tools for showing the hardness of a PCSP is the theorem proved in \cite{Barto21:algebraic} stating that if $\Pol(\bA,\bB)$ has \emph{bounded essential arity}, i.e., every function from $\Pol(\bA,\bB)$ only depend on a fixed number of variables and $\Pol(\bA, \bB)$ does not have a constant function then $\PCSP(\bA,\bB)$ is NP-complete. This implies that if there is a minion homomorphism from $\Pol(\bA_1,\bB_1)$ to $\Pol(\bA_2,\bB_2)$ and $\Pol(\bA_2,\bB_2)$ has bounded essential arity and does not have a constant function, then $\PCSP(\bA_1,\bB_1)$ is also NP-complete. This method was  expanded to abstract minors in the place of $\Pol(\bA_2,\bB_2)$ \cite{Brakensiek20:power,Ciardo23:CLAP,Ciardo22:Sherali} and further strengthened in~\cite{Zin21}.


\smallskip
\noindent\textbf{Topology.}
One of the most interesting ways to apply the approach outlined above is through topology. It was first introduced in~\cite{Kro2019,Kro2020} and then further developed in~\cite{Wro2020}. As this method has been our motivation, we describe the main steps.\\[2mm]
\textit{Step 1. Simplicial complexes.}
A \emph{simplicial complex} is a downward-closed family $K$ of non-empty sets, called the \emph{faces} of $K$. The elements of the faces are \emph{vertices} of $K$, the set of all vertices is denoted by $V(K)$. A \emph{$\mathbb Z_2$-simplicial complex} is a simplicial complex endowed with a permutation $-$ of order 2 on $V(K)$ such that for every face $\sg$ the set $-\sg$ is also a face. If $G$ is a graph, the set of vertices of the corresponding simplicial complex $\Hom(K_2,G)$ (in the notation of \cite{Kro2020}) is the set of edges of $G$, and the faces are   (the edge sets of) complete bipartite subgraphs of $G$. 

Polymorphisms of a $\zZ_2$-simplicial complex $K$ to a $\zZ_2$-simplicial complex $K'$ are defined to be mappings $V(K)^n\to K'$ that map a Cartesian product of faces to a face and preserve the action of $-$. The set of all polymorphisms from $K$ to $K'$ is a minion denoted $\Pol(K,K')$. The crucial fact is that for any graphs $G,H$ there exists a minion homomorphism from $\Pol(H,G)$ to $\Pol(\Hom(K_2,H),\Hom(K_2,G))$. \\[2mm]
\textit{Step 2. Topological spaces.}
In the next step simplicial complexes are converted into topological spaces. A $\zZ_2$-topological space is a topological space $\mathcal X$ with a continuous mapping $-$ on it such that $-(-v)=v$ for every $v\in\mathcal X$. If, for a simplicial complex $K$, $V(K)=\{\vc vn,-v_1,\dots -v_n\}$, then the \emph{geometric realization} $|K|$ of $K$ is a subset of $\zR^n$ constructed as follows. Every $\pm v_i$ is represented by a vector $(0,\dots,0,\pm1,0,\dots,0)$, where the 1 is in the $i$-th position. For every face $\sg$ of $K$, the convex hull of the vectors corresponding to the vertices of $\sg$ is a set in $|K|$. Finally, $-$ acts on $|K|$ by flipping the signs of the vectors.

It is again possible to define polymorphisms from a $\zZ_2$-space $\mathcal X$ to a $\zZ_2$-space $\mathcal Y$ as a continuous mapping from $\mathcal X^n$ to $\mathcal{Y}$ preserving the action of $-$. And again the set of all polymorphisms $\Pol(\mathcal{X},\mathcal{Y})$ is a minion. This time it is not  true that for two simplicial complexes $K,K'$ there is a minion homomorphism from $\Pol(K,K')$ to $\Pol(|K|,|K'|)$, but a similar property holds: polymorphisms from $\Pol(|K|,|K'|)$ have to be treated up to homotopy. \\[2mm]
%
%
\textit{Step 3. Fundamental groups.}
Fundamental groups of topological spaces are an indispensable tool in algebraic topology, \cite{All2001}. The fundamental groups of some spaces such as the 1-dimensional sphere $\cS^1$, that is, a circle, are known. The fundamental group $\pi_1(\cS^1)$ is isomorphic to $\zZ$ that can also be viewed as a free group with one generator. 

For groups $G_1,G_2$ a polymorphism $f$ from $G_1$ to $G_2$ is a mapping $f:G_1^n\to G_2$ that commutes with the group operation, that is, $f(a_1\cdot b_1,\ldots, a_n\cdot b_n)=f(\vc an)\cdot f(\vc bn)$. The set $\Pol(G_1,G_2)$ of all polymorphisms from $G_1$ to $G_2$ is a minion, and for any spaces $\mathcal{X},\mathcal{Y}$ there exists a minion homomorphism from $\Pol(\mathcal{X},\mathcal{Y})$ to $\Pol(\pi_1(\mathcal{X}),\pi_1(\mathcal{Y}))$.\\[2mm]
\indent
Starting from graphs $H,G$ by composing all the minion homomorphisms in the three steps one obtains a homomorphism from $\Pol(H,G)$ to \linebreak $\Pol(\pi_1(|\Hom(K_2,H)|),\pi_1(|\Hom(K_2,G)|))$. If the graphs $H,G$ satisfy some extra conditions, the minion of group polymorphisms is isomorphic to $\Pol(\zZ,\zZ)$. The latter minion satisfies the condition of bounded essential arity, and therefore one concludes that $\PCSP(H,G)$ is NP-complete.

\smallskip
\noindent\textbf{Our results.}
In this paper we attempt to develop a discrete analog of the topological approach outlined above by replacing topological structures with combinatorial constructions and some basic group-theoretic concepts. Firstly, we try to isolate the essential components of the approach in order to make as general as possible, in particular by extending the technique to relational structures that have more than one predicate and the predicates are not necessarily binary, making the connection between relational structures and simplicial complexes less rigid, and by relaxing the requirements on $\zZ_2$-complexes. Some elements of these have occurred in the literature, see, e.g.\ \cite{Filakovsky23:hardness}, however, in a fairly rudimentary way. Secondly, we try to make the approach more elementary by removing several stages from it, and effectively minimizing the use of topology so that it is more accessible for those without a background in topology (ourselves included).

First, let  $\cH$ a simplicial complex. We follow \cite[Chapter 4]{Singer1967} to introduce edge-path groups of $\cH$. The set of vertices of $\cH$ is converted into a graph $G(\cH)$, in which two tuples are connected with an edge if they belong to a same face of $\cH$. Fixing a tuple $t_1$ we consider \emph{cycles} in $G(\cH)$ starting at $t$ that are essentially closed walks in that graph. Let $C=\vc tm,t_1$, $t_1=t$, be a cycle and consider the following transformations of $C$: removing the vertex $t_i$ if $t_{i-1},t_i,t_{i+1}$ all belong to the same face, and its reverse transformation. Two cycles are \emph{homotopic} if they can be obtained from each other by a sequence of transformations. The set of all cycles homotopic to the cycle $C$ is denoted by $[C]$ and called the \emph{homotopy class} of $C$. The set of all homotopy classes with the operation of concatenation of cycles forms a group denoted $\pi_1(\cH, t_1)$ and called the \emph{edge-path group of $\cH$ at $t_1$}. 

Another group related to $\cH$ is constructed as follows: Let $\cT$ be a spanning tree of the connected component $D$ of $G(\cH)$ containing $t_1$. By $\zG^{\cH,\cT}$ we denote the finitely presented group with the edges from $D$ as generators and a set of relations defined by the faces of $\cH$, see Section~\ref{sec:homotopic}. Then by \cite[p.\ 134, Theorem 16]{spanier1989algebraic}, \cite[Chapter 4]{Singer1967} $\zG^{\cH,\cT}$ and $\pi_1(\cH, t_1)$ are isomorphic, and thus $\pi_1(\cH, t_1)$ is a finitely presented group.

In the next step we replace the last step of  the approach from \cite{Kro2020} with the following statement\footnote{While all the current application of Theorem~\ref{the:free-group-intro} use the case $m=1$, we prove it in as general form as possible. In Theorem~\ref{free group bounded essential arity} we prove a more general statement.}.

\begin{theorem} \label{the:free-group-intro}
Let $(\bA,\bB)$ be a PCSP-template and $\mathcal{F}$ a free group with finitely many generators. If there exists a minion homomorphism $\varphi: \Pol(\bA,\bB) \rightarrow \Pol(\mathbb Z^m, \mathcal{F})$, $m\in\mathbb N$, then there exists a number $N \in \mathbb{N}$ such that for any $f \in \Pol(\bA,\bB)$, the essential arity of $\varphi(f)$ is bounded by $N$.
\end{theorem}

In order to apply Theorem~\ref{the:free-group-intro} we find a free group $\mathcal{F}$ and identify sufficient conditions for the existence of a nontrivial minion homomorphism $\varphi$. Let $(\bA,\bB)$ be a PCSP-template, $\rel^\bA,\rel^\bB$ a pair of relations primitive positive definable in $\bA,\bB$ by the same primitive positive formula, and $\cH^\bA,\cH^\bB$ a pair of simplicial complexes with vertices from $\rel^\bA,\rel^\bB$, whose faces are preserved by polymorphisms from $\Pol(\bA,\bB)$. The following statement in some cases gives us a free group.


\begin{lemma} \label{lem: conditions on A to be a free group-intro}
The group $\pi_1(\cH^\bB, t_1)$ is free if any two maximal faces of $\cH^\bB$ have at most one common element.
\end{lemma}

For the existence of a minion homomorphism we introduce the property of \emph{non-degeneracy} formulated in terms of cycles in $\cH^\bA,\cH^\bB$, see Section~\ref{sec:hardness}. Combining the two conditions we obtain the main result of the paper.

\begin{theorem}\label{Main Theorem-intro}
Let $(\textbf A, \textbf B)$ be a non-degenerate template. If $\pi_1(\cH^\bB, t)$ is a free group, then $\PCSP(\textbf A, \textbf B)$ is NP-hard. 
\end{theorem}

The condition of non-degeneracy is not very natural, so we also offer a weaker, but more natural condition that is similar to $\zZ_2$-structures from \cite{Kro2020}, but applies to a much wider class of structures, see Theorem~\ref{Main Theorem 2}. Finally, in Section~\ref{sec:applications} we give some applications of the main theorems.

We believe that the results of this paper are only the first step in exploring the potential of groups of this kind. One possible direction is to consider groups similar to homotopy groups of higher order. 

\section{Preliminaries}\label{sec:prelims}


We use the notation $[n] = \{ 1, 2, \ldots, n\}$ throughout the paper.

\vspace*{3mm}

\noindent\textbf{Relational structures and homomorphism.}
A \textit{relational structure} is a tuple $\textbf{A}=(A;R_1^{\textbf{A}}, R_2^{\textbf{A}}, \ldots, R_w^{\textbf{A}})$ where each $R_i^{\textbf{A}} \subseteq A ^ {ar(R_i)}$ is a relation of arity $ar(R_i) \geq 1$. In this article, all relational structures are finite, that is, their domains are finite. Note that we denote the domain (universe) of $\textbf A$ by $A$. Two relational structures $\textbf{A}=(A;R_1^{\textbf{A}}, R_2^{\textbf{A}}, \ldots, R_w^{\textbf{A}})$ and $\textbf{B}=(B;R_1^{\textbf{B}}, R_2^{\textbf{B}}, \ldots, R_w^{\textbf{B}})$ are similar if they have the same number of relations and $ar(R_i ^ {\textbf{A}}) = ar(R_i ^ {\textbf{B}})$ for each $i \in [w]$.

For two similar relational structures $\textbf A$ and $\textbf B$, a \textit{homomorphism} from $\textbf A$ to $\textbf B$ is a mapping $h: A \rightarrow B$ such that, for each $i$ if $(a_1, \ldots, a_{ar(R_i)}) \in R_i ^ {\textbf{A}}$ then $(h(a_1), \ldots, h(a_{ar(R_i)})) \in R_i ^ {\textbf{B}}$. We write $\textbf {A} \rightarrow \textbf{B}$ to denote that a homomorphism from $\textbf A$ to $\textbf B$ exists and $\textbf {A} \not\rightarrow \textbf{B}$ otherwise.

\vspace*{3mm}

\noindent\textbf{CSP and PCSP.}
For a structure $\bA$, $\CSP(\textbf{A})$ is defined to be the following problem:
\begin{itemize}
\item \textbf{Input:}
A relational structure $\textbf I$ similar to $\textbf A$.
\item \textbf{Output:}
The output is YES if there is a homomorphism from $\textbf {I}$ to $\textbf{A}$ and NO if there is no homomorphism from $\textbf {I}$ to $\textbf{A}$.
\end{itemize}

The Promise CSP is defined in a similar way.

\begin{definition}
Let $(\textbf A, \textbf B)$ be a pair of similar fixed relational structures such that $\textbf {A} \rightarrow \textbf{B}$ (such a pair is called a \emph{PCSP template}). The problem $\PCSP(\textbf{A},\textbf{B})$ is defined as follows:
\begin{itemize}
\item \textbf{Input:}
A structure $\textbf I$ similar to both $\textbf A$ and $\textbf{B}$ such that either $\textbf {I} \rightarrow \textbf{A}$ or $\textbf {I} \not\rightarrow \textbf{B}$.
\item \textbf{Output:}
The output is YES if there is a homomorphism from $\textbf {I}$ to $\textbf{A}$ and NO if there is no homomorphism from $\textbf {I}$ to $\textbf{B}$.
\end{itemize}
\end{definition}

Many combinatorial problems can be expressed as a PCSP. In particular, $\CSP(\bA)=\PCSP(\bA,\bA)$.

\begin{example}
For fixed natural numbers $c \geq k \geq 1$, $\PCSP(\textbf K_k, \textbf K_c)$ is the problem of deciding whether a given graph $\textbf G$ is $k$-colorable or not even $c$-colorable (note that we are promised that the input graph $\textbf G$ is either $k$-colorable or not even $c$-colorable). This problem is known as the \textit{Approximate Graph Coloring Problem} (AGCP).
The AGCP has been conjectured to be NP-hard for any $c \geq k \geq 3$, but this is still open in many cases. The current state-of-the-art for NP-hardness is $c=2k-1$ for $k =3,4,5$~\cite{Bul2019} and $c= {k \choose {\lfloor \frac{k}{2} \rfloor}} -1$ for $k > 5$~\cite{Wro2020}. 
\end{example}

\begin{example}
$\PCSP(\textbf G, \textbf H)$ for graphs $\textbf G$ and $\textbf H$ with $\textbf G\to\textbf H$ is called the \textit{Approximate Graph Homomorphism Problem} (AGHP). This problem has been conjectured to be NP-hard for any non-bipartite graphs $\textbf G$ and $\textbf H$~\cite{Bra2021}. The case $\textbf G= \textbf C_{2k+1}$ (odd cycles), $\textbf H=\textbf K_3$ has been shown to be NP-hard~\cite{Kro2019, Kro2020}. The authors in~\cite{Kro2020} used some topological techniques to prove the NP-hardness of $\PCSP(\textbf C_{2k+1}, \textbf K_3)$. It was also recently proved in \cite{Avvakumov25} that $\PCSP(\textbf G, \textbf H)$ is NP-hard for any pair of non-bipartite 4-colorable graphs $G,H$.
\end{example}


\noindent\textbf{Polymorphisms, minions, and primitive-positive definitions.}
A mapping $f:A^n \rightarrow B$ is called a \textit{polymorphism} of $\PCSP(\textbf{A},\textbf {B})$ of arity $n$, denoted by $ar(f)=n$, if for any $i \in [w]$ and any sequence of tuples
\[
(r^1_1,\ldots,r^1_{ar(R_i ^ {\textbf{A}})}),\ldots,(r^n_1,\ldots,r^n_{ar(R_i ^ {\textbf{A}})})
\]
in $R_i^{\textbf{A}}$, we have
$(f(r^1_1,\ldots,r^n_1),\ldots,f(r^1_{ar(R_i ^ {\textbf{A}})},\ldots,r^n_{ar(R_i ^ {\textbf{A}})})) \in R_i^{\textbf{B}}$.
The set of all polymorphisms of $\PCSP(\textbf{A},\textbf {B})$ is denoted by $\Pol(\textbf{A},\textbf {B})$ and the set of all polymorphisms of arity $n$ is denoted by $\Pol(\textbf A, \textbf B)^{(n)}$. We define $\Pol(\textbf A)$ to be $\Pol(\textbf A, \textbf A)$.

It has been proved that the complexity of $\PCSP(\textbf{A},\textbf {B})$ depends only on the set of all polymorphism of $\PCSP(\textbf{A},\textbf {B})$~\cite{Bra2018,Bra2021}, that is, if $\Pol(\textbf A, \textbf B) \subseteq \Pol(\textbf C, \textbf D)$ then $\PCSP(\textbf C, \textbf D)$ is reducible to $\PCSP(\textbf A, \textbf B)$.


An $m$-ary function $g:A^m \rightarrow B$ is called a \textit{minor} of an $n$-ary function $f:A^n \rightarrow B$ if there exists a mapping $\pi:[n] \rightarrow [m]$ such that
$g(x_1,\ldots,x_m)=f(x_{\pi(1)},\ldots,x_{\pi(n)})$
for any $x_1,\ldots,x_m \in A$. In such a case, we write $g = f^{\pi}$ or $f \xrightarrow{\pi} g$.

For two sets $A$ and $B$, let $\mathcal{O}(A,B)=\{f: A^n \rightarrow B \,| \, n \geq 1 \}$ be the set of all functions from $A^n$ to $B$ for $n \geq 1$. A (\textit{function}) \textit{minion} $\mathcal{M}$ on a pair of sets $(A, B)$ is a non-empty subset of $\mathcal{O}(A, B)$ that is closed under taking minors. That is, if $f$ is an $n$-ary function in $\mathcal{M}$, then for any $m \in \mathbb{N}$ and any map $\pi:[n] \rightarrow [m]$, $f^{\pi}$ is also a function in $\mathcal{M}$. Let $\mathcal{M}^{(n)}$ denote the set of $n$-ary functions in $\mathcal{M}$ for a fixed $n \geq 1$.

Our principal example of a minion is the set of polymorphisms of a PCSP template. It is not hard to see that for a template $(\bA,\bB)$ the set $\Pol(\textbf A, \textbf B)$ is a minion. Indeed, if $f$ is a polymorphism, then $f^\pi$ is also a polymorphism for any map $\pi: [n] \rightarrow [m]$.

Another example of a minion is given by the set of all group polymorphisms. Let $\textbf G$, $\textbf H$ be two groups. Then a function $f: G^n \rightarrow H$ is called a \emph{polymorphism} if
$f(a_1 \cdot b_1,\ldots,a_n \cdot b_n)=f(a_1,\ldots,a_n) \cdot f(b_1,\ldots,b_n)$ for all $\forall a_i,b_i \in \bold G$.
In other words $f$ is a group homomorphism from $G^n$ to $H$. The set of all polymorphisms is a minion denoted by $\Pol(\textbf G, \textbf H)$.  

Let $\mathcal{M}_1, \mathcal{M}_2$ be two function minions (not necessarily on the same pair of sets). Then a mapping $\xi: \mathcal{M}_1 \rightarrow \mathcal{M}_2$ is called a \textit{minion homomorphism} if the following conditions are satisfied.
\begin{enumerate}
\item 
For any $n \in \mathbb{N}$, it maps an $n$-ary function $f \in \mathcal{M}_1$ to an $n$-ary function in $\mathcal{M}_2$, i.e., $ar(f)=ar(\xi(f))$.
\item 
It preserves minors, i.e., for any $f \in \mathcal{M}_1^{(n)}$ and $\pi: [n] \rightarrow [m]$ we have $\xi(f)^{\pi}=\xi(f^{\pi})$, i.e.,
$$\xi(f)(x_{\pi(1)},\ldots,x_{\pi(n)})=\xi(f^\pi)(x_1,\ldots,x_m).$$
\end{enumerate}
The notions of a minion and minion homomorphism can be extended to \emph{abstract minions}, in which $\mathcal M$ is an arbitrary set partitioned into subsets $\mathcal M^{(n)}$, $n\in\mathbb N$, with properly defined minor operation. In this paper we almost exclusively use function minions.


A homomorphism between minions of polymorphisms of two PCSP templates gives rise to a reduction between PCSPs.

\begin{theorem}[\cite{Bul2019}, Theorem 3.1]\label {thr: reduction, minion homomorphism}
Let $(\textbf A, \textbf B)$ and $(\textbf C, \textbf D)$ be $\PCSP$ templates. If there is a minion homomorphism $\xi: \Pol(\textbf A, \textbf B) \rightarrow \Pol(\textbf C, \textbf D)$, then $\PCSP(\textbf C, \textbf D)$ is log-space reducible to $\PCSP(\textbf A, \textbf B)$.
\end{theorem}


Let $\bA$ be a relational structure with the base set $A$. A \emph{primitive-positive (pp-)formula} $\Phi(\vc xr)$ is a first-order formula that only uses existential quantifiers, conjunction, predicates of $\bA$, and the equality relation. The $r$-ary relation $\rel$ consisting of all tuples $(\vc ar)$ such that $\Phi(\vc ar)$ is true is said to be \emph{pp-definable} in $\bA$. For a PCSP template a pp-formula $\Phi$ defines a pair of relations $\rel^\bA,\rel^\bB$ in each of the structures $\bA,\bB$. The property of pp-definable relations we need (\cite[Theorem~2.26]{Barto21:algebraic}, see also \cite{Bra2021,Pippenger02:Galois}) is that every polymorphism of $(\bA,\bB)$ is also a polymorphism of $((A;\rel^\bA),(B;\rel^\bB))$.

\vspace*{3mm}

\noindent\textbf{Hardness condition.} 
A coordinate $i$ of a function $f: A^n \rightarrow B$ is called \textit{essential} if there exist $a_1,\ldots,a_n \in A$ and $b_i \in A$ such that
$$
f(a_1,\ldots, a_{i-1},a_i,a_{i+1},\ldots,a_n) \neq f(a_1,\ldots, a_{i-1},b_i,a_{i+1},\ldots,a_n).
$$
A minion $\mathcal{M}$ is said to have a \textit{bounded essential arity} if there exists a constant number $k$ such that any function $f$ in $\mathcal{M}$ has at most $k$ essential coordinates.

\begin{theorem}{\rm \cite[Proposition 5.14]{Bul2019}}\label{Minion Theorem}
Let $(\bA, \bB)$ be a template. Assume that there exists a minion homomorphism $\xi: \Pol(\bA, \bB) \rightarrow \mathcal{M}$ for some minion $\mathcal{M}$ on a pair of (possibly infinite) sets such that $\mathcal{M}$ has bounded essential arity and does not contain a constant function (i.e., a function without essential coordinates). Then $\PCSP(\bA, \bB)$ is NP-hard.
\end{theorem}

As an example of an application of Theorem~\ref{Minion Theorem} consider the problem $\PCSP(\textbf K_3, \textbf K_4)$. It is possible to construct a minion homomorphism $\xi: \Pol(\textbf K_3, \textbf K_4) \rightarrow \mathcal{P}_{\{0,1\}}$, where $\mathcal{P}_{\{0,1\}}$ is the minion of projections on $\{0,1\}$. Any function $f \in \mathcal{P}_{\{0,1\}}$ has exactly $1$ essential coordinate (which also implies $f$ is not constant), so by Theorem~\ref{Minion Theorem} $\PCSP(\textbf K_3, \textbf K_4)$ is NP-hard.

If we have a minion homomorphism $\xi: \Pol
(\bA, \bB) \rightarrow \mathcal{M}$, then the image of $\xi$, denoted by $\img_{\xi}(\mathcal{M})$, is also a minion and $\xi^\prime: \Pol(\bA, \bB) \rightarrow \img_{\xi}(\mathcal{M})$ where $\xi^\prime(f):= \xi(f)$ is a minion homomorphism. This implies the following corollary.

\begin{corollary} \label {col: Minion Homomorphism}
Let $(\bA, \bB)$ be a template. Assume that there exists a minion homomorphism $\xi: \Pol(\bA, \bB) \rightarrow \mathcal{M}$ for some minion $\mathcal{M}$ on a pair of (possibly infinite) sets such that $\img_\xi(\mathcal{M})$ has bounded essential arity and does not contain a constant function (i.e., a function without essential coordinates). Then $\PCSP(\bA, \bB)$ is NP-hard.
\end{corollary}

\vspace*{3mm}

\noindent\textbf{Finitely presented groups.}
We will also use some basic concepts of finitely presented groups \cite{Rotman95:groups}. A \emph{presentation} of a group $\bG$ consists of a set $S$ of generators, and a set $R$ of relations among those generators. A presentation is finite if both $S$ and $R$ are finite. Then elements of $\bG$ are represented by \emph{group words} that is formal products of generators and their inverses. A relation is an equality of the form $w_1=w_2$, where $w_1,w_2$ are group words. Relations can be used to transform group words by replacing some occurrences of $w_1$ or $w_1^{-1}$ with $w_2,w_2^{-1}$, respectively, or the other way round. Words that can be transformed to one another represent the same element of $\bG$. A \emph{free group} is a group presented by $S$ and $R$, where $R=\emptyset$. A word in a free group $\bF$ is reduced if it does not contain a subword $aa^{-1}$ or $a^{-1}a$ for a generator $a$. Every element in $\bF$ is equal to a reduced word and all reduced words are different.
A free abelian group is an example of a free group with relations $ab=ba$ for any pair of generators $a,b \in S$.

\section{Groups and PCSP templates}\label{sec:homotopic}
In this section we study how analogues of homotopy groups can be introduced for a relational structure. While \cite{Kro2020} first construct a simplicial complex from a relational structure and then study its properties and how it gives rise to a topological space, we start with an arbitrary simplicial complex (or, equivalently, a hypergraph) on a relation pp-definable in a relational structure and try to extract the conditions essential for our purposes. 

\subsection{Edge-path groups}\label{sec:edge-path}

First, we remind the notion of the \emph{edge-path group} of a simplicial complex, see \cite{Singer1967}. Let $\cH=(V,E)$ be a simplicial complex (a hypergraph). The graph $G(\cH)$ of $\cH$ is defined as follows: the vertex set of $G(\cH)$ is $V$ and $u,v\in V$ are connected with an edge if and only if $s,t\in A$ for some $A\in E$. We consider paths and cycles in $G(\cH)$. A cycle $t_1, \ldots, t_k$, $t_1=t_k$,  is called a \textit{null cycle} if $k=2$. Two paths $P_1= t_1,\ldots, t_k$ and $P_2= t^\prime_1,\ldots, t^\prime_r$ are said to be \textit{homotopic}, denoted $P_1 \sim P_2$, if it is possible to transform $P_1$ to $P_2$ by a sequence of transformations of the following two types (called \textit{expansions} and \textit{contractions}): (a) If an element $t\in V$ belongs to the same face as $t_i$, $t_{i+1}$ for some $i \in [k-1]$, then replace $P_1$ by $t_1,\ldots,t_i,t,t_{i+1}, \ldots, t_k$; and (b) if $t_i$, $t_{i+1}$, and $t_{i+2}$ belong to the same face for some $i \in [k-2]$, then replace $P_1$ by $t_1,\ldots,t_i,t_{i+2}, \ldots, t_k$.

Note that if two paths $P_1$ and $P_2$ are homotopic, they have the same initial and terminal elements, i.e., $t_1=t^\prime_1$ and $t_k=t^\prime_r$ since expansions and contractions do not change the initial and terminal elements. A cycle $C$ is said to be \textit{null-homotopic} if it is homotopic to a null cycle.

Given two cycles $C_1=t_1,\ldots,t_k,t_1$ and $C_2=t^\prime_1, \ldots, t^\prime_r,t^\prime_1$ starting at the same element $t_1=t^\prime_1$, we define $C_1 \cdot C_2$ to be the concatenation of them, i.e., $C_1 \cdot C_2=t_1,\ldots,t_k,t_1,t^\prime_2, \ldots, t^\prime_r,t_1$. The \textit{homotopy class} of a cycle $C$ is the set of cycles that are homotopic to $C$. We denote the homotopy class of a cycle $C$ by $[C]$. 
Multiplication on the set of homotopy classes is defined in the natural way: If $C_1=t_1,\ldots,t_k,t_1$ and $C_2=t^\prime_1, \ldots, t^\prime_r,t^\prime_1$ with $t_1=t^\prime_1$, then $[C_1] \cdot [C_2] := [C_1 \cdot C_2]$. It is not hard to see (cf. \cite[Chapter 4]{Singer1967}) that the set of homotopy classes along with the operation $\cdot$ is a group denoted $\pi_1(\cH, t_1)$. Note that for any $t'$ from the same connected component of $G(\cH)$ as $t_1$, the groups $\pi_1(\cH, t_1)$ and $\pi_1(\cH, t')$ are isomorphic. So, we will often use $\pi_1(\cH)$ rather than $\pi_1(\cH, t_1)$, at least when $G(\cH)$ is connected.
(Note that the definitions of homotopic, null-homotopic and $\pi_1$ are similar to the ones in algebraic topology~\cite{All2001}.)

%
%

Next, we introduce another group related to $\cH$. This group is finitely presented and we will then show that the two groups are isomorphic. Fix a vertex $t_1$ of $G(\cH)$ and let $\mathcal{T}$ be a spanning tree of the connected component of $G(\cH)$ containing $t_1$. We call such a tree a \textit{maximal tree} containing $t_1$ (if $G(\cH)$ is connected, then any maximal tree would be a spanning tree of $G(\cH)$). By $V(\cT)$ and $E(\cT)$ we denote the vertex set and the edge set of $\cT$, respectively. Consider the finitely presented group, denoted by $\mathbb{G}^{\cH, \mathcal{T}}$, whose generators are the (oriented) edges of the connected component of $G(\cH)$ that contains $t_1$, the identity element is denoted by $\mathbbm{1}$, and the set of relations is: 
\begin{itemize}
\item[(G1)] 
If $e \in E(\mathcal{T})$, then $e=e^{-1}=\mathbbm{1}$;
\item[(G2)] 
if $t \in V(\mathcal{T})$, then $(tt) = \mathbbm{1}$;
\item[(G3)] 
for elements $\{t_1, t_2, t_3 \} \in V(\mathcal{T})$ that belong to the same face, it holds that $(t_1t_2)\cdot (t_2t_3)\cdot(t_3t_1)= \mathbbm{1}$. In particular, $(t_1t_2)\cdot (t_2t_3)=(t_3t_1)^{-1}$.
\end{itemize}
Note that relation (G3) also implies that $(t_1 t_2) \cdot (t_2 t_1) = \mathbbm{1}$. Indeed, if $t_3=t_1$, then 
$(t_1 t_2) \cdot (t_2 t_1)=(t_1 t_2)\cdot (t_2 t_1)\cdot (t_1t_1)= \mathbbm{1}$ by (G3).

The mapping $\rho^{t_1,\mathcal{T}}$ from $\pi_1(\cH, t_1)$ to $\mathbb{G}^{\cH, \mathcal{T}}$ is defined as follows. Let $C= t_1,t_2, \ldots, t_k, t_1$ be a cycle. Then we define $\rho^{t_1,\mathcal{T}}([C]):=(t_1t_2)\cdot (t_2t_3)\ldots (t_{k-1}t_k)\cdot (t_kt_1)$. The following lemma can be found in \cite[Section 4, Theorem 5]{Singer1967}, for the sake of completeness we also give a proof.

\begin{lemma}[\cite{Singer1967}] \label{isomorphism lemma}
Let $\cH=(V,E)$ be a simplicial complex. For any element $t_1\in V$ and any maximal tree $\mathcal{T}$ containing $t_1$, the mapping $\rho^{t_1,\mathcal{T}}: \pi_1(\cH, t_1) \rightarrow \mathbb{G}^{\cH, \mathcal{T}}$ is a group isomorphism.
\end{lemma}

\begin{proof}
First of all, $\rho^{t_1,\mathcal{T}}$ is a well-defined mapping. In order to demonstrate this let $C= t_1,t_2, \ldots, t_k, t_1$ be a cycle, and a cycle $C^\prime$ is produced from $C$ by expansion or contraction. We need to prove that $\rho^{t_1,\mathcal{T}}([C]) = \rho^{t_1,\mathcal{T}}([C^\prime])$.\\[1mm] 
{\bf Expansion.} If $t_i$, $t_{i+1}$, and $t$ belong to the same face, then 
\[
C^\prime=t_1, \ldots, t_i, t, t_{i+1}, \ldots, t_k.
\]
So we have 
\[
\rho^{t_1,\mathcal{T}}([C^\prime])=(t_1t_2)\cdot\ldots\cdot(t_{i-1}t_i)\cdot (t_i t)\cdot(t t_{i+1})\cdot(t_{i+1} t_{i+2})\cdot\ldots\cdot(t_kt_1),
\] 
which is equal to $\rho^{t_1,\mathcal{T}}([C])$ since $(t_i t_{i+1}) = (t_i t) \cdot (t t_{i+1})$ by (G3).\\[1mm] 
{\bf Contraction.} This case is similar.\\[1mm]

It is straightforward that $\rho^{t_1,\mathcal{T}}$ is a group homomorphism since 
$$
\rho^{t_1,\mathcal{T}}([C_1] \cdot [C_2]) = \rho^{t_1,\mathcal{T}}([C_1 \cdot C_2]) = \rho^{t_1,\mathcal{T}}([C_1]) \cdot \rho^{t_1,\mathcal{T}}([C_2]).
$$

We now find a homomorphism $\gamma^{t_1,\mathcal{T}}: \mathbb{G}^{\cH, \mathcal{T}} \rightarrow \pi_1(\cH, t_1)$ and prove that $\rho^{t_1,\mathcal{T}} \circ \gamma^{t_1,\mathcal{T}}$ and $\gamma^{t_1,\mathcal{T}} \circ \rho^{t_1,\mathcal{T}}$ are identity maps which will complete the proof of the lemma. Let $a= (v_1v_2)\cdot(v_3v_4)\cdot\ldots\cdot(v_{2k-1}v_{2k})$ be an element of $\mathbb{G}^{\cH, \mathcal{T}}$. Then for any $i \in [k-1]$, there exists a unique shortest path $P_i$ from $v_{2i}$ to $v_{2i+1}$ in $\mathcal{T}$. Let also $P_0$ and $P_k$ be unique shortest paths in $\mathcal{T}$ from $t_1$ to $v_1$ and from $v_{2k}$ to $t_1$, respectively. Then
$$
C = P_0, v_1, v_2, P_1, v_3, v_4 \ldots, P_i, v_{2i+1}, v_{2i+2}, \ldots, 
P_{k-1}, v_{2k-1}, v_{2k}, P_{k}
$$
is a cycle in $\GA$ starting at $t_1$. We set $\gamma^{t_1,\mathcal{T}}(a) = [C]$. To prove that $\gamma^{t_1,\mathcal{T}}$ is a well-defined map, we need to show that if $b=(u_1u_2) \cdot(u_3 u_4)\cdot\ldots\cdot(u_{2r-1}u_{2r})$ is equal to $a$ in $\mathbb{G}^{\cH, \mathcal{T}}$, then $\gamma^{t_1,\mathcal{T}}(a)=\gamma^{t_1,\mathcal{T}}(b)$. Equivalently, it suffices to show that if we apply one of the three relations (G1)--(G3) to $a$, then the image with $\gamma^{t_1,\mathcal{T}}(a)$ does not change. The details are described below:
\begin{itemize}
\item 
For any $xy\in E(\mathcal{T})$, the element $b = (v_1v_2)\cdot\ldots\cdot(v_{2i-1}v_{2i})\cdot (xy)\cdot(v_{2i+1}v_{2i+2})\cdot\ldots\cdot(v_{2k-1}v_{2k})$, which is equal to $a$ in $\mathbb{G}^{\cH, \mathcal{T}}$, has the same image under $\gamma^{t_1,\mathcal{T}}$, i.e., $\gamma^{t_1,\mathcal{T}}(a) = \gamma^{t_1,\mathcal{T}}(b)$. Indeed, let $P_x$ and $P_y$ be the shortest paths in $\mathcal{T}$ form $v_{2i}$ to $x$ and from $y$ to $v_{2i+1}$, respectively. Then $\gamma^{t_1,\mathcal{T}}(b)=[C^\prime]$ where
\begin{align*}
C^\prime & = P_0, v_1, v_2, P_1, v_3, v_4 \ldots, P_{i-1}, v_{2i-1}, v_{2i}, P_x, x,y, P_y, v_{2i+1},\\
& \hspace*{15mm} v_{2i+2} \ldots, P_{k-1}, v_{2k-1}, v_{2k}, P_{k}.
\end{align*}
Now, if we can prove that any two paths $Q_1$ and $Q_2$ in $\mathcal{T}$ from $v_{2i}$ to $v_{2i+1}$ are homotopic, then we are done with the proof of this part since the only difference between $C$ and $C^\prime$ is that $P_i$ is replaced with $P_x, x,y, P_y$ which are both paths from $v_{2i}$ to $v_{2i+1}$ in $\cT$.

\begin{claim}
For tuples $u, v \in V(\mathcal{T})$, let $Q_1, Q_2$ be paths in $\mathcal{T}$ from $u$ to $v$. Then $Q_1$ and $Q_2$ are homotopic.
\end{claim}
\begin{claimproof}
Let $Q$ be the shortest path from $u$ to $v$ in $\mathcal{T}$, then we prove that $Q_1$ is homotopic to $Q$. Suppose $Q \neq Q_1$. By a sequence of contractions, we convert $Q_1$ to $Q$. Since $\mathcal{T}$ is a tree, either $Q_1$ has two equal consecutive tuples, i.e., $Q_1=u,\ldots, z,z, \ldots, v$, or there is an edge of $\cT$ that in $Q_1$ is traversed forward and then immediately backward, i.e., $Q_1=u,\ldots, z,t,z, \ldots, v$. In the former case, we can remove one of the $z$'s by contraction. In the latter case, we can remove $t$ by contraction. As the shortest path is unique in $\cT$, repeating these transformations we transform $Q_1$ into $Q$. Hence, $Q_1$ and $Q$ are homotopic. Similarly, $Q_2$ and $Q$ are homotopic; thus $Q_1$ and $Q_2$ are homotopic which completes the proof of the claim.
\end{claimproof}

\item 
For any tuple $x$, the element $b = (v_1v_2)\cdot\ldots(v_{2i-1}v_{2i})\cdot (xx)\cdot (v_{2i+1}v_{2i+2})\cdot\ldots(v_{2k-1}v_{2k})$, which is equal to $a$ in $\mathbb{G}^{\cH, \mathcal{T}}$, has the same image with $\gamma^{t_1,\mathcal{T}}$, i.e., $\gamma^{t_1,\mathcal{T}}(a) = \gamma^{t_1,\mathcal{T}}(b)$. To see this it suffices to set $x=y$ in the previous case.

\item 
In the final case assume $v_{2i}=v_{2i+1}$ and $v_{2i-1},v_{2i},v_{2i+2}$ belong to the same face for some $i\in[k-2]$. Then we need to consider the element
$$
b = (v_1v_2)\cdot\ldots\cdot(v_{2i-3}v_{2i-2})\cdot(v_{2i-1}v_{2i+2})\cdot(v_{2i+3}v_{2i+4})\cdot\ldots\cdot(v_{2k-1}v_{2k}),
$$
which is equal to $a$ in $\mathbb{G}^{\textbf A, \mathcal{T}}$ (since $(v_{2i-1}v_{2i})\cdot (v_{2i+1}v_{2i+2}) = (v_{2i-1}v_{2i+2})$ by (G3)). Then the path $P_i$ is trivial, and so $v_{2i}$ can be contracted out of $C$. Therefore, $\gamma^{t_1,\mathcal{T}}(a) = \gamma^{t_1,\mathcal{T}}(b)$.
\end{itemize}
The fact that both $\rho^{t_1,\mathcal{T}} \circ \gamma^{t_1,\mathcal{T}}$ and $\gamma^{t_1,\mathcal{T}} \circ \rho^{t_1,\mathcal{T}}$ are identity maps can be easily checked. Hence, $\rho^{t_1,\mathcal{T}}$ is an isomorphism.
\end{proof}

\subsection{Edge-path group of a relational structure}\label{sec:edge-path-structure}

Another ingredient of the approach in \cite{Kro2020} is the claim that faces of simplicial complexes defined on a PCSP template (graphs in \cite{Kro2020}) are preserved under polymorphisms. We introduce a more general version of this connection. 

Let $(\bA,\bB)$ be a PCSP template and $\Phi$ a pp-formula with $r$ free variables in the vocabulary of $(\bA,\bB)$. Let $\rel^\bA,\rel^\bB$ be the $r$-ary relations defined by $\Phi$ in $\bA,\bB$, respectively. Let also $\cH^\bA,\cH^\bB$ be simplicial complexes on $\rel^\bA,\rel^\bB$, respectively. The pair $\cH^\bA,\cH^\bB$ is said to be \emph{polymorphism stable} if for any polymorphism $f(\vc xn)$ of $(\bA,\bB)$ and any faces $\vc\sg n$ of $\cH^\bA$ the set 
$
f(\sigma_1, \ldots, \sigma_n):= \{ f(t_1,\ldots,t_n) \; | \; t_i \in \sigma_i \}
$
is a face of $\cH^\bB$, where $f(t_1,\ldots,t_n)$ is computed coordinate-wise.

Next, we describe a construction that gives an example of how polymorphism-stable simplicial complexes can be obtained. In particular, it includes all the existing applications of the topological method in PCSP. In addition to $\Phi$ and $\rel$ let $\Psi$ be another pp-formula in the vocabulary of $(\bA,\bB)$, with $2r$ free variables, such that $\Psi$ defines relations $\relo^\bA,\relo^\bB$ in $\bA,\bB$, respectively, and such that $\relo^\bA,\relo^\bB$ are binary reflexive and symmetric relations on $\rel^\bA,\rel^\bB$, respectively. A binary reflexive and symmetric relation is often called a \emph{tolerance} and we will be using this term from now on. In particular, for $\rel$ and $\relo$ as above, we will say that  $\relo$ is a tolerance of $\rel$. The following is an example of a tolerance when $\bA, \bB$ are graphs.

\begin{example}\label{exa:forks}
Let $(\bA,\bB)$ be a PCSP template where $\bA,\bB$ are graphs. Let $\rel^\bA=E(\bA),\rel^\bB=E(\bB)$, i.e., the formula $\Phi$ is just $E(x,y)$. For $\Psi$ consider $\Psi(x_1,y_1,x_2,y_2)= (y_1=y_2)$. Then $\relo^\bA,\relo^\bB$ consist of pairs $((s_1,t_1),(s_2,t_2))$ of $\bA,\bB$, respectively, with $t_1=t_2$.
\end{example}

If $\relo$ is a tolerance on a set $A$, a \emph{class} of $\relo$ is a subset $B\sse A$ such that $(a,b)\in\relo$ for any $a,b\in B$. A \emph{maximal class} of $\relo$ is defined as a maximal set like this in terms of inclusion. Unlike equivalence classes, maximal classes of a tolerance may overlap. 

\begin{lemma}\label{lem:tol-classes}
Let $\rel$ be a relation pp-definable in a PCSP template $(\bA,\bB)$ and $\relo$ its tolerance. Let also $\cH^\bA,\cH^\bB$ denote simplicial complexes on $\rel^\bA,\rel^\bB$, respectively, whose faces are classes of the tolerances $\relo^\bA,\relo^\bB$ respectively. Then the pair $\cH^\bA,\cH^\bB$ is polymorphism-stable.
\end{lemma}

\begin{proof}
It suffices to show that for any $t_1,\ldots,t_n,t'_1,\ldots,t'_n\in\rel^\bA$ such that $(t_i,t'_i)\in\relo^\bA$, $i\in[n]$, it holds that $(f(t_1,\ldots,t_n),f(t'_1,\ldots,t'_n))\in\relo^\bB$. This however is straightforward as $\relo$ is pp-definable in $(\bA,\bB)$.
\end{proof}

The approach of \cite{Kro2020} uses a specific definition of faces that we are going to state now, again in a slightly more general setting, and show that it fits into our framework. 

Let $\bA=(A,\rel^\bA)$ be relational structure with a single $r$-ary relation. A \emph{face} of the relation $R^{\bA}$ is a set $\sg\sse R^\bA$ of the form $\sg\sse X_1\tms X_r\sse\rel^\bA$ for some $\vc Xr\sse A$.  We now introduce a tolerance whose classes are exactly the faces of $\cH^\bA$.

For $i\in[r]$, let $\relo_i$ denote a generalization of the tolerance introduced in Example~\ref{exa:forks}, that is, $(a,b)\in\relo_i$ if and only if $a_j=b_j$ for $j\in[r]-\{i\}$ and $a=(\vc ar),b=(\vc br)\in\rel$. This tolerance is defined by the pp-formula
$Q_i(x,y)=\bigwedge_{j\in[r]-\{i\}}(x_j=y_j)
$.
Next, let $\pi$ be a permutation of $[r]$ and the binary relation $\relo_\pi$ is given by $\relo_\pi=\relo_{\pi(1)}\circ\dots\circ\relo_{\pi(r)}$. This can be defined by the pp-formula
\[
\relo_\pi(x,y)=\exists z^1\zd z^{r-1}\ \ \relo_{\pi(1)}(x,z^1)\wedge\relo_{\pi(2)}(z^1,z^2)\wedge\dots\wedge \relo_{\pi(r)}(z^{r-1},y).
\]
Note that $\relo_\pi$ is \emph{not} necessarily a tolerance, as it may not be symmetric. The following lemma shows however that the intersection of $\relo_\pi$ for all permutations $\pi$ is a tolerance, and its classes are the faces of $\cH^\bA$. 

\begin{lemma}\label{lem:tolerance}
Let $\bA$ and $\relo_\pi$ be as above. Then the relation 
\[
\relo(x,y)=\bigwedge_{\pi \text{ is a permutation of } [r]} \relo_\pi
\]
is a tolerance, for any $\vc Xr\sse A$ such that $X_1\tm\dots\tm X_r\sse\rel$, the set $X_1\tm\dots\tm X_r$ is a class of $\relo$, and every maximal class of $\relo$ has this form.
\end{lemma}

\begin{proof}
Since each of $\relo_i$ is reflexive, so are $\relo_\pi$ and $\relo$. To check that $\relo$ is symmetric it suffices to check that $(b,a)\in\relo_\pi$ whenever $(a,b)\in\relo$ for any permutation $\pi$ of $[r]$. As $(a,b)\in\relo$, $(a,b)\in \relo_{\pi'}$, where $\pi'(i)=\pi(r-i+1)$, that is, 
\[
(a,b)\in\relo_{\pi'}=\relo_{\pi'(1)}\circ\dots\circ\relo_{\pi'(r)}=\relo_{\pi(r)}\circ\dots\circ\relo_{\pi(1)}.
\]
This implies $(b,a)\in\relo_\pi$.

Now, let $\vc Xr\sse A$ be such that $X_1\tm\dots\tm X_r\sse\rel$ and $a,b\in X_1\tm\dots\tm X_r$. The sequence 
\[
a=(\vc ar),(b_1,a_2\zd a_r)\zd(\vc b{r-1},a_r),(\vc br)=b
\]
witnesses that $(a,b)\in\relo_\pi$ for the identity permutation $\pi$. For other permutations the argument is essentially the same. 

Finally, let $B\sse\rel$ be a maximal class of $\relo$ and $X_i=\{a_i\mid(\vc ar)\in B\}$.
Clearly, $B\sse X_1\tm\dots\tm X_r$, we show that $X_1\tm\dots\tm X_r\sse B$. It suffices to prove that for any $a=(\vc ar)\in B$, any $i\in[r]$, and any $b_i\in X_i$, the tuple $(a_1\zd a_{i-1},b_i,a_{i+1}\zd a_r)$ belongs to $B$. Without loss of generality let $i=1$. As $b_1\in X_1$, there is $b=(\vc br)\in B$, in particular, $(a,b)\in\relo_\pi$ for the identity permutation $\pi$. In other words, $(a,b)\in\relo_1\circ\dots\circ\relo_r$ and there are $a=a^0,a^1\zd a^{r-1},a^r=b\in\rel$ such that $(a^{j-1},a^j)\in\relo_j$ for $j\in[r]$. The tuples $a^{j-1}$ and $a^j$ can only differ in $j$th coordinate. Therefore, for $a^1=(a^1_1\zd a^1_r)$ we have $a^1_1=b_1$ and $a^1_j=a_j$ for $j\in\{2\zd r\}$, as required.
\end{proof}

\subsection{Polymorphisms and edge-path groups}\label{sec:poly-homo}

In this section we study the connection between polymorphisms of a PCSP template $(\bA,\bB)$ and edge-path groups of relations pp-definable in $\bB$. We fix a relation $\rel$ pp-definable in $(\bA,\bB)$, and a polymorphism-stable pair $\cH^\bA,\cH^\bB$ of simplicial complexes on $\rel^\bA,\rel^\bB$ such that $\GB$ is connected. Note that in this case $\pi_1(\cH^\bB,t)$ does not depend on the choice of $t\in\rel^\bB$.  

An \textit{$f$-path} in $\GB$ is a sequence of the form 
$$
f(t^1_1,\ldots,t^1_n), f(t^2_1,\ldots,t^2_n), \ldots, f(t^m_1,\ldots,t^m_n)
$$
for tuples $t^i_j$, $i \in [m], j \in [n]$, such that $t^i_j t^{i+1}_j$ is an edge of $\GA$ for every $i \in [m-1], j \in [n]$. Each $f(t^i_1,\ldots,t^i_n)$ is called a vertex, and $t^i_j$ is called the $j$-th coordinate of the $i$-th vertex. 

In the sequel $f$-paths will be constructed as follows. Let $C= t_1, \ldots, t_k, t_1$ be a cycle in $\GA$. Then in the notation above every $t^i_j$ belongs to $\{t_1, \ldots, t_k\}$, and for each $i\in[m]$, if $t^i_j=t_s$ then $t^{i+1}_j\in\{t_{s-1},t_s,t_{s+1}\}$ (Note that $t_0 = t_k$ and $t_{-1}=t_{k-1}$). An $f$-path satisfying these conditions for a cycle $C$ will be called \emph{$C$-bound}. A $C$-bound $f$-path is said to be \emph{semi-normal} if it satisfies the extra condition: for every $i\in[m-1]$ the sequences $(t^i_1,\ldots,t^i_n)$ and $(t^{i+1}_1,\ldots,t^{i+1}_n)$ differ in the set $S_i\sse[n]$ of coordinates, and for each $j\in S_i$, $t^i_j=t_{s_j},t^{i+1}_j=t_{s_j+1}$ or for each $j\in S_i$, $t^i_j=t_{s_j},t^{i+1}_j=t_{s_j-1}$. An $f$-path is called a \textit{normal} if it is semi-normal and $|S_i|=1$ for every $i\in[m]$.  

The notation for normal and semi-normal $f$-paths can be simplified. Instead of writing the $i$-th vertex in the full format, we can only indicate the difference, that is, the set $S_i$, between the current vertex and the previous vertex, and the direction of the change by $S^+_i$ or $S^-_i$. If the $f$-path is normal and $S_i=\{s\}$, we write $s^+,s^-$ instead of $S^+_i,S^-_i$. For instance,  the normal $f$-path
$$
C_1= f(t_1, t_1, \ldots, t_1), f(t_2, t_1, \ldots, t_1), \ldots, f(t_k, t_1, \ldots, t_1), f(t_1, t_1, \ldots, t_1)
$$
can also be represented as
$$
C_1= f(t_1, t_1, \ldots, t_1), 1^+, 1^+, \ldots, 1^+,  f(t_1, t_1, \ldots, t_1).
$$
Note that we write the first and last vertices in the full format. Similarly, the semi-normal $f$-path
$$
C= f(t_1, t_1, \ldots, t_1), f(t_2, t_2, \ldots, t_2), \ldots, f(t_k, t_k, \ldots, t_k), f(t_1, t_1, \ldots, t_1)
$$
can also be represented as
$$
C= f(t_1, t_1, \ldots, t_1), [n]^+, [n]^+, \ldots, [n]^+,  f(t_1, t_1, \ldots, t_1).
$$

By Lemma~\ref{lem:tol-classes} any semi-normal $f$-path (cycle) in $\GB$ is a path (cycle) in $\GB$. The next lemma shows that certain rearrangements of semi-normal paths are homotopic.

\begin{lemma}\label{lem:cycles-body}
In the above notation\\[2mm]
(1) $C_i$, $i\in[n]$, is a cycle in $\GB$ that starts at $t$.\\[1mm]
(2) For any $i,j\in[n]$, $[C_i] \cdot [C_j] = [C_j] \cdot [C_i]$.\\[1mm]
(3) Let $\ov t=(t_{i_1}\zd t_{i_n})$.
For any $S_1,S_2\sse[n]$, $S_1\cap S_2=\emptyset$ the cycles 
\begin{align*}
& D=f(t_1\zd t_1)\zd f(\ov t),(S_1\cup S_2)^+\zd f(t_1\zd t_1),\\
& D'=f(t_1\zd t_1)\zd f(\ov t),S_1^+,S_2^+\zd f(t_1\zd t_1), \text{and}\\
& D''=f(t_1\zd t_1)\zd f(\ov t),S_2^+,S_1^+\zd f(t_1\zd t_1)
\end{align*}
are homotopic (here dots denote the common part of $D,D',D''$).
\end{lemma}

\begin{proof}
As is easily seen $v_{j-1}, v_j,v_{j+1}$ belong to the same face so we can contract away $v_j$ from $P_1$ to obtain $P_3$. So $P_1$ is homotopic to $P_3$. In a similar way $P_2$ is homotopic to $P_3$.
\end{proof}

If $P= v_1, \ldots, v_{j-1}, S^+, v_{j+1}, \ldots, v_m$ is an $f$-path where $S= \{ a_1, \ldots, a_s \}$, then by Lemma~\ref{lem:cycles-body}(3) $P$ is homotopic to
$$
P^\prime= v_1, \ldots, v_{j-1}, a_1^+, \ldots, a_s^+, v_{j+1}, \ldots, v_m.
$$
This is also true for any permutation $b_1, \ldots, b_s$ of $a_1, \ldots, a_s$. So we have the following statement.

\begin{corollary} \label{col:f-path 2}
Let $P_1= v, a_1^+, \ldots, a_s^+, v$ and $P_2= v, b_1^+, \ldots, b_s^+, v$ be two normal $f$-paths where $b_1, \ldots, b_s$ is a permutation of $a_1, \ldots, a_s$. Then $P_1$ and $P_2$ are homotopic.
\end{corollary}


We complete this section with a statement relating polymorphisms of $(\bA,\bB)$ with polymorphisms involving the homotopy group of $\cH^\bB$. 

\begin{theorem} \label{the:minion homomorphism from polymorphism}
The mapping $\xi_C: \Pol(\bA,\bB) \rightarrow \Pol(\mathbb Z, \pi_1(\GB))$ constructed above is a minion homomorphism.
\end{theorem}

\begin{proof}
To prove that $\xi_C$ is a minion homomorphism, we need to show that for any ($n$-ary) $f\in\Pol(\bA,\bB)$: (a) $\xi_C(f)\in \Pol(\mathbb Z, \pi_1(\GB))$, and (b) $\xi_C(f^\pi) = \xi_C(f)^\pi$ for any map $\pi: [n] \rightarrow [m]$. Item (a) is straightforward from Lemma~\ref{lem:cycles-body}(1,2), as it implies, in particular, that the cycles $[C_1]\zd[C_n]$ generate an Abelian subgroup of $\pi_1(\GB)$, and therefore $\xi_C(f)$ is a homomorphism of $\zZ^n$ to $\pi_1(\GB)$. For item (b), since $\xi_C(f)$ is a polymorphism, the equality
$$
(\xi_C(f)^\pi)_i = [C_{j_1}] \cdot [C_{j_2}] \cdots [C_{j_r}]
$$
holds for any $i \in [m]$, where $S= \{ j_1, \ldots, j_r \}\sse[n]$ is the set of numbers whose images by $\pi$ are equal to $i$. Let
\begin{align*}
&C^\pi_i= P, ff(t_1,\ldots, t_1), i^+\zd i^+, f(t_1,\ldots, t_1), P^{-1},\\
&C^\prime= P, f(t_1, t_1, \ldots, t_1), S^+, \ldots, S^+, f(t_1, t_1, \ldots, t_1), P^{-1}
\end{align*}
be cycles in $\GB$. Observe that, in fact, $C_i^\pi$ and $C^\prime$ are the same cycle in $\GB$ since $f^\pi(x_1,\ldots,x_m) = f(x_{\pi(1)}, \ldots, x_{\pi(n)})$. By Lemma~\ref{lem:cycles-body}(3), $C^\prime$ is homotopic to
$$
C^\prime_1= P, f(t_1, t_1, \ldots, t_1), j_1^+, \ldots, j_r^+, \ldots, j_1^+, \ldots, j_r^+, f(t_1, t_1, \ldots, t_1), P^{-1},
$$
where each $S^+$ is replaced with the sequence $j_1^+, \ldots, j_r^+$ in $C^\prime$. By applying Lemma~\ref{lem:cycles-body}(3) multiple times, $C^\prime_1$ is homotopic to
\begin{align*}
C^\prime_2 &= P, f(t_1, t_1, \ldots, t_1), j_1^+, \ldots, j_1^+, \ldots, j_r^+, \ldots, j_r^+, f(t_1, t_1, \ldots, t_1), P^{-1} \\
&\sim C_{j_1} \cdot C_{j_2} \cdots C_{j_r},
\end{align*}
where the number of the $j_i^+$'s is $k$ for any $i \in [r]$. Since $P \cdot P^{-1}$ and $P^{-1} \cdot P$ are both null-homotopic, we have
$$
(\xi_C(f)^\pi)_i = [C_{j_1}] \cdot [C_{j_2}] \cdots [C_{j_r}] = [C_{j_1} \cdot C_{j_2} \cdots C_{j_r}] = [C^\prime_2].
$$
Since $C^\prime_2$, $C^\prime_1$, $C^\prime$, and $C_i^\pi$ are homotopic, we can conclude that $(\xi_C(f)^\pi)_i = [C^\pi_i]$. On the other hand, $\xi_C(f^\pi)_i = [C^\pi_i]$ by the definition of $\xi_C(f^\pi)_i$. Hence $(\xi_C(f)^\pi)_i = \xi_C(f^\pi)_i$ for any $i \in [m]$ which implies $\xi_C(f)^\pi = \xi_C(f^\pi)$ completing the proof.
\end{proof}

We say that a group $\bF$ is \emph{Abelian-bounded} if it satisfies the following two conditions:
\begin{itemize}
    \item[(A1)]
    Every Abelian subgroup of $\bF$ is a finitely generated free Abelian group (i.e.\ $\zZ^n$ for some $n$).\footnote{As is easily seen, this condition is equivalent to the following one: $\bF$ is torsion-free and every its Abelian subgroup is finitely generated. However, in order to obtain the results below it cannot be relaxed by dropping the requirement that Abelian subgroups are finitely generated.}
    \item[(A2)]
    Every nontrivial Abelian subgroup is contained in finitely many maximal Abelian subgroups.
\end{itemize}

\begin{lemma}\label{lem:free-abelian-bounded}
Any finitely generated free group is Abelian-bounded.
\end{lemma}

\begin{proof}
Let $\bF$ be a finitely generated free group. By Nielsen-Schreier theorem \cite{Rotman95:groups} every subgroup of a free group is free. As the only nontrivial Abelian free group is 1-generated, (A1) follows. To prove (A2) let $\zA$ be a nontrivial Abelian subgroup of $\bF$ generated by an element $s\in\bF$. Let $|s|$ denote the length of the reduced word representing $s$ as a product of generators. If we prove that for every $s_1\in\bF$ such that $s_1$ generates a subgroup $\zB$ and $\zA\sse\zB$, it holds that $|s_1|<|s|$, we get that there are only finitely many Abelian subgroups $\zB$ with $\zA\sse\zB$, and (A2) follows.

Let $\zA,\zB$ and $s,s_1$ as before, we may assume that $s_1$ is a reduced word. Since $\zA\sse\zB$, $s$ belongs to $\zB$ and is a power of $s_1$ in $\bF$, say, $s=s_1^k$ for some $k$. As $s$ is an arbitrary generator of $\zA$, we may assume $k\ge0$. Let $s_1=y s'_1 y^{-1}$ be such that $|y|$ is as large as possible, and let $s'_1=b_1 \cdots b_\ell$, $\ell > 0$. Then by the choice of $y$ we have $b_\ell\ne b_1^{-1}$. As is easily seen, $s=s_1^k$ can be represented as $y (s'_1)^k y^{-1}$, and this representation is reduced. Thus, $|s|=k\ell+2|y|$, while $|s_1|=\ell+2|y|$. If $\zA\ne\zB$, then $k\ge2$ and we obtain $|s|>|s_1|$, as required. 
\end{proof}

We are in a position to prove the second main result of this section. For vectors $\ov a,\ov b\in\zZ^k$, $\ov a=(\vc ak),\ov b=(\vc bk)$, let $|\ov a|$ denote $|a_1|+\dots+|a_k|$ and let $\ov a\ov b$ denote $a_1b_1+\dots+a_kb_k$.

\begin{theorem}\label{free group bounded essential arity}
Let $\mathcal{M}$ be an abstract minion such that $\mathcal M^{(n)}$ is finite for all $n$, and assume that $\phi: \mathcal{M} \rightarrow \Pol(\mathbb Z^k, \bF)$ is a minion homomorphism where $\bF$ is an Abelian-bounded group. There exist numbers $P,N \in \mathbb{N}$ such that for any $n\in\mathbb N$ and any $f \in \mathcal{M}^{(n)}$ such that $\phi(f) \neq 0$, there exist (commuting) elements $\vc sr \in \bF$, $r\le P$, $\ov c_{ij} \in \mathbb Z^k$, $i\in[n],j\in[r]$, such that 
\[
\phi(f)(\ov x_1,\ldots, \ov x_n)=\prod_{j=1}^r s_j^{\ov c_{1j} \ov x_1+ \cdots + \ov c_{nj} \ov x_n}
\]
and $|\ov c_{11}|+\cdots+|\ov c_{nr}| \leq N$.
\end{theorem}

\begin{proof}
First, observe that any polymorphism $h\in\Pol(\zZ^k,\bF)$ is a group homomorphism from $\zZ^{nk}$ to $\bF$, where $n$ is the arity of $h$. This means the range $\zA_h=h(\zZ^{nk})$ of $h$ is an Abelian subgroup of $\bF$. For $f\in\cM$ we will use $\zA_f=\zA_{\phi(f)}$. Therefore, for any $f \in \mathcal{M}^{(n)}$, there exist (commuting) elements $\vc sq \in \bF$ and $\ov c_{ij} \in \mathbb Z^k$, $i\in[n],j\in[t]$, such that $\phi(f)(\ov x_1,\ldots, \ov x_n)=\prod_{j=1}^q s_j^{\ov c_{1j} \ov x_1+ \cdots + \ov c_{nj} \ov x_n}$, and we need to prove a bound on $\sum_{i,j} |\ov c_{ij}|$ and $q$.

Next, we consider the set $\mathcal{M}^{(1)}$. Let $f\in\cM^{(1)}$; by what is shown above $\phi(f)(\ov x)=s_1^{\overline{c_1} \, \ov x}\cdots s_q^{\overline{c_q} \, \ov x}$, where $\vc sq$ generate an Abelian subgroup $\zA_f$ of $\bF$. Note that this representation is not unique and we want to pick one or a finite set of such representations. By condition (A2) there are finitely many maximal Abelian subgroups $\vc{\zB^f}{t_f}$ of $\bF$ such that $\zA_f\sse\zB^f_i$, $i\in[t_f]$. Let $\cB=\{\zB^f_i\mid f\in\cM^{(1)}, i\in[t_f]\}$ be the set of all such subgroups. Fix a set $B^f_i$ of generators for each subgroup $\zB^f_i\in\cB$. 

\smallskip

{\sc Claim.}
For any $f\in\cM$, $\zA_f$ is a subgroup of a member of $\cB$.

\smallskip

Let $g=f^\pi\in\cM^{(1)}$ be the unary minor of $f$. Then clearly $\zA_g$ is a subgroup of $\zA_f$. Hence, every maximal Abelian subgroup of $\bF$ containing $\zA_f$ also contains $\zA_g$, and therefore belongs to $\cB$, which completes the proof of the claim. \hfill$\Box$

\smallskip

Now we can identify $N$ and $P$. Let $\ell$ be the maximal number of generators of subgroups in $\cB$ and $P=2^{\ell k}$. For each $f\in\cM^{(n)}, n \leq P$ and $\zB\in\cB$ such that $\zA_f\sse\zB$, let $\phi(f)(\ov x_1,\ldots,\ov x_n)=\prod_{j=1}^q s_j^{\ov c_{1j} \ov x_1+ \cdots + \ov c_{nj} \ov x_n}$ be a representation of $\phi(f)$ for the canonical generating set $B=\{\vc sq\}$ of $\zB$. Then set 
\begin{align*}
N_{f,\zB} &=|\ov c_{11}|+\dots+|\ov c_{nq}|,\\ 
N_f &=\max\{N_{f,\zB}\mid \zB\in\cB, \zA_f\sse\zB\},\quad \text{and} \\
N &=\max\{N_f\mid f\in\cM^{(n)}, n \leq P\}.
\end{align*}

Let $f \in \mathcal{M}^{(n)}$, $\zB\in\cB$ such that $\zA_f\sse\zB$, and let $\phi(f)(\ov x_1,\ldots,\ov x_n)=\prod_{j=1}^r s_j^{\ov c_{1j} \ov x_1+ \cdots + \ov c_{nj} \ov x_n}$, $\ov c_{ij} \in \mathbb Z^k$, be the representation of $\phi(f)$ over the canonical basis of $\zB$. we show that this representation satisfies the conditions of the theorem. Clearly, $r \leq \ell < P$. If $n\le P$ then $|\ov c_{11}|+\dots+|\ov c_{nr}|\le N$ by the choice of $N$.
Now, we define a map $\pi: [n] \rightarrow [2^{rk}]$ in such a way $\pi(i)=\pi(i')$ if and only if for any $j \in [r], m \in [k]$, $c_{i,j}^m$ and $c_{i',j}^m$ have the same sign. More formally, let $\pi: [n] \rightarrow [2^{rk}]$ be a mapping such that $\pi(i)=(a_0, \ldots, a_{rk-1})$ where $(a_0, \ldots, a_{rk-1})$ is the representation of $\pi(i)$ base 2 and $a_j=1$ if $c_{i, \lfloor \frac{j}{k} \rfloor +1}^{(j \bmod k) + 1} \geq 0$ and $a_j=0$ otherwise. Then let $g \in \mathcal{M}^{(2^{rk})}$ be given by $g(\ov x_1\zd \ov x_{2^{rk}})=f(\ov x_{\pi(1)},\ldots,\ov x_{\pi(n)})$.
As $\phi$ preserves minors, 
$$
\phi(g)(\ov x_1\zd \ov x_{2^{rk}})=\phi(f)(\ov x_{\pi(1)},\ldots,\ov x_{\pi(n)}).
$$
The representation of $\phi(g)$ in the basis $\{\vc sr\}$ is
\[
\phi(g)=\prod_{j=1}^r s_j^{\ov{c'}_{1j} \ov x_1+ \cdots + \ov{c'}_{2^{rk}j} \ov x_{2^{rk}}},\qquad \text{where}\quad \ov{c'}_{ij}=\sum_{\pi(u)=i}\ov c_{uj}. 
\]
For each $i\in[2^{rk}], j\in[r], v \in [k]$ all the $\ov c_{uj}^v$ such that $u \in [n], \pi(u)=i$ have the same sign. Therefore for any $i \in [2^{rk}], j \in [r]$,
\[
|\ov{c'}_{ij}|=\left|\sum_{\pi(u)=i} \ov c_{uj}\right|=\sum_{\pi(u)=i}|\ov c_{uj}|
\]
implying
\[
\sum_{u\in[n],j\in[r]}|\ov c_{uj}|=\sum_{i\in[2^{rk}],j\in[r]}|\ov{c'}_{ij}|\le N.
\]
\end{proof}


\section{Complexity of PCSP and edge-path groups}\label{sec:hardness}

In this section we apply the results from previous sections to the complexity of the PCSP. 

Let again $(\bA,\bB)$ be a PCSP template, and fix a relation $\rel$ pp-definable in $(\bA,\bB)$ and a polymorphism-stable pair $\cH^\bA,\cH^\bB$ of simplicial complexes on $\rel^\bA,\rel^\bB$ such that $\GB$ is connected. If there are $t\in\rel^\bB$ and a cycle $C=t_1\zd t_m,t_1$ in $\GA$ such that for any homomorphism $g:\bA\to\bB$ the cycle $P,g(C),P^{-1}$, where $P$ is a path from $t$ to $g(t_1)$, is not null-homotopic, the template $(\bA,\bB)$ is said to be \emph{non-degenerate with respect to $\rel,C$}. (Note that as $\GB$ is connected, the vertex $t$ can be chosen arbitrarily.)

We start with an auxiliary lemma.

\setcounter{theorem}{32}
\begin{lemma} \label{non-constant minion homomorphism}
Let $(\textbf A, \textbf B)$ be a non-degenerate template with respect to $R, t, C$. Let
$$
\xi_C: \Pol(\textbf A, \textbf B) \rightarrow \Pol(\mathbb{Z}, \pi_1(\GB))
$$
be the minion homomorphism defined before Theorem~\ref{the:minion homomorphism from polymorphism}. Then for any $f \in \Pol(\textbf A, \textbf B)$ of arity $n$, $\xi_C(f)$ is not constant.
\end{lemma}

\begin{proof}
We only need to prove this lemma for unary functions in $\Pol(\textbf A, \textbf B)^{(1)}$. This is because there exists a unary function $g \in \Pol(\textbf A, \textbf B)^{(1)}$ such that $g(x)=f(x,\ldots,x)$ for any $x \in A$. Since $g$ is a minor of $f$ and $\xi_C$ is a minion homomorphism, we have $\xi_C(g)(x)=\xi_C(f)(x,\ldots,x)$ for any $x \in \mathbb Z$. If $\xi_C(g)(x) \neq \mathbbm{1}$ for some $x \in\mathbb Z$, then $\xi_C(f)(x,\ldots,x) \neq \mathbbm{1}$. 

Now, suppose $C = t_1, t_2, \ldots, t_k, t_1$ and let $f \in \Pol(\textbf A, \textbf B)^{(1)}$ be a unary function, note that it is a homomorphism from $\bA$ to $\bB$. By the assumption the cycle $C'= P, f(t_1), f(t_2), \ldots, f(t_k), f(t_1), P^{-1}$ where $P$ is the path from $t$ to $f(t_1)$ and $P^{-1}$ is the reverse of $P$, is not null-homotopic. By the definition of $\xi_{C}$ before Theorem~\ref{the:minion homomorphism from polymorphism}, $\xi_{C}(f)(1)=[C']$. As $C'$ is not null-homotopic, $\xi_{C}(f)(1) \neq \mathbbm{1}$, as required.

\end{proof}

\begin{lemma} \label{free group structure}
Let $\bF$ be a free group with finitely many generators and let $f: \mathbb Z^{nm} \rightarrow \bF$ be a function in $\Pol(\mathbb Z^m, \bF)^{(n)}$. Then there exist integers $c_{ij}\in \mathbb Z$, $i\in[n],j\in[m]$, and an element $s \in \bF$ such that 
\[
f(\ov x_1,\ldots,\ov x_n)=s^{\sum_{i\in[n],j\in[m]}c_{ij} x_{ij}},
\]
where $\ov x_i=(x_{i1}\zd x_{im})\in\zZ^m$.
\end{lemma}

\begin{proof}
First, note that it suffices to prove the lemma for $m=1$. By the definition of group polymorphism, we have
\[
f(x_1,x_2, \ldots , x_n)=f(1,0,\ldots,0)^{x_1} \cdot f(0,1,\ldots,0)^{x_2} \cdots f(0,0,\ldots,1)^{x_n},
\]
as
\[
(x_1,x_2, \ldots , x_n)=x_1\cdot(1,0,\ldots,0)+\dots+x_n\cdot(0,0,\ldots,1).
\]
Since $\mathbb Z^n$ is Abelian and $f$ is a homomorphism, the image $f(\mathbb Z^n)$ in $\bF$ is an Abelian subgroup of $\bF$. By Nielsen-Schreier theorem \cite{Rotman95:groups} every subgroup of a free group is free. As the only Abelian free group is 1-generated, $f(\mathbb Z^n)$ is a 1-generated free group. The result follows.
\end{proof}

\begin{lemma} \label{non-constant minion homomorphism2}
Let $(\textbf A, \textbf B)$ be a non-degenerate template with respect to $R, t, C$, and let $\pi_1(\cH^\bB)$ be a free group. Let also
$$
\xi_C: \Pol(\textbf A, \textbf B) \rightarrow \Pol(\mathbb{Z}, \pi_1(\GB))
$$
be the minion homomorphism defined before Theorem~\ref{the:minion homomorphism from polymorphism}. Then for any $f \in \Pol(\textbf A, \textbf B)$ of arity $n$, there exist integers $c_{i}\in \mathbb Z$, $i\in[n]$, and an element $s \in \pi_1(\cH^\bB)$ such that 
\[
\xi_C(f)(x_1,\ldots,x_n)=s^{\sum_{i\in[n]}c_{i} x_{i}},
\]
where $x_i\in\zZ$, $|c_1|+\cdots+|c_n| \neq 0$ and $s \neq \mathbbm{1}$.
\end{lemma}

\begin{proof}
By Lemma~\ref{free group structure}, integers $c_{i}\in \mathbb Z$, $i\in[n]$, and an element $s \in \pi_1(\cH^\bB)$ with the required properties exist. Suppose that $|c_1|+\cdots+|c_n| = 0$ or $|s| = 0$. Then, $\xi_C(f)(x_1,\ldots,x_n)=\mathbbm{1}$, which contradicts Lemma~\ref{non-constant minion homomorphism}.
\end{proof}

We are now ready to state and prove our main hardness theorem.

\begin{theorem}\label{Main Theorem}
Let $(\bA,\bB)$ be a non-degenerate template with respect to $\rel,C$. If $\pi_1(\GB)$ is an Abelian-bounded group, then $\PCSP(\bA,\bB)$ is NP-hard.
\end{theorem}

\begin{proof}
Let $\xi_C: \Pol(\textbf A, \textbf B) \rightarrow \Pol(\mathbb{Z}, \pi_1(\GB))$ be the minion homomorphism defined before Theorem~\ref{the:minion homomorphism from polymorphism} and let $f \in \Pol(\textbf A, \textbf B)$ be any polymorphism. By Lemma~\ref{non-constant minion homomorphism2}, $\xi_C(f)$ is not constant, and by Theorem~\ref{free group bounded essential arity}, it has bounded essential arity. So, $\PCSP(\textbf A, \textbf B)$ is NP-hard by Corollary~\ref{col: Minion Homomorphism}.
\end{proof}

We now propose properties of the template $(\bA,\bB)$ that guarantee that the conditions of Theorem~\ref{Main Theorem} hold. 

\begin{lemma} \label{lem: conditions on A to be a free group}
The group 
$\pi_1(\cH^\bB)$ is free if any two maximal faces of $R^\bB$ have at most one common vertex.
\end{lemma}

\begin{proof}
We use the group $\mathbb{G}^{\GB,\mathcal{T}}$ for a spanning tree $\mathcal T$ of $\GGB$.
Let $\sigma_1, \ldots, \sigma_n$ be the maximal faces in $\GB$ and let $T^{\sigma_i}$ be the set of edges in $\sigma_i$ that also belong to $E(\mathcal{T})$. Let $E^{\sigma_i}$ be some set of edges in $\sigma_i$ such that $E^{\sigma_i} \cup T^{\sigma_i}$ forms a spanning tree in $\sigma_i$. We now prove that $\mathbb{G}^{\GB,\mathcal{T}}$ is a free group with the set of generators $E^{\sigma_1} \cup \cdots \cup E^{\sigma_n}$.
Indeed, let $\bF$ be a free group whose generators are $(E^{\sigma_1} \cup T^{\sigma_1}) \cup \cdots \cup (E^{\sigma_n} \cup T^{\sigma_n})$ such that all generators in $T^{\sigma_1} \cup \cdots \cup T^{\sigma_n}$ are the identity element. Moreover, if $uv \in (E^{\sigma_1} \cup T^{\sigma_1}) \cup \cdots \cup (E^{\sigma_n} \cup T^{\sigma_n})$ is a generator, then $(uv) = (vu)^{-1}$.
To prove the lemma, we need to find an isomorphism $\eta: \mathbb{G}^{\GB,\mathcal{T}} \rightarrow \bF$ (The proof idea is similar to the proof of Lemma~\ref{isomorphism lemma}).

For an edge $uv$ in $\sigma_i$, we define $\eta(uv)$ to be $(v_1 v_2) \cdot (v_2 v_3) \cdot \ldots \cdot (v_{k-1} v_k)$ where $P: u=v_1,v_2, \ldots, v_k=v$ is the unique shortest path form $u$ to $v$ in the spanning tree $E^{\sigma_i} \cup T^{\sigma_i}$. First, we need to prove that $\eta$ is a well-defined map. Note that since any two maximal faces of $\GB$ have at most one common tuple, $\sg_i$ is the only maximal face containing $uv$ and can be used to define $\eta(uv)$.  Next, it suffices to show that if we apply one of the three relations (G1)--(G3) to any element $a \in \mathbb{G}^{\GB,\mathcal{T}}$, then the image $\eta(a)$ does not change. The relations (G1) and (G2) are obvious since generators in $T^{\sigma_1} \cup \cdots \cup T^{\sigma_n}$ are the identity. Relation (G3) can be easily obtained from the following fact.
For an edge $uv \in \sigma_i$ and any two (not necessarily shortest) paths $P:u=a_1, \ldots, a_k=v$ and $Q:u=b_1, \ldots, b_r=v$ from $u$ to $v$ in $E^{\sigma_i} \cup T^{\sigma_i}$, the elements $(a_1a_2) \cdot (a_2a_3) \cdot \ldots \cdot (a_{k-1}a_k)$ and $(b_1b_2) \cdot (b_2b_3) \cdot \ldots \cdot (b_{r-1}b_r)$ are equal in $\bF$. So $\eta$ is a well-defined map. Moreover 
$\eta$ is a homomorphism from $\mathbb{G}^{\GB, \mathcal{T}}$ to $\bF$ since $\eta(a \cdot b) = \eta(a) \cdot \eta(b)$ for any $a,b \in \mathbb{G}^{\GB, \mathcal{T}}$.

We now find a homomorphism $\eta^{-1}: \bF \rightarrow \mathbb{G}^{\GB, \mathcal{T}}$ and then prove that both $\eta \circ \eta^{-1}$ and $\eta^{-1} \circ \eta$ are identity maps. For any generator $uv$ in $\bF$ we know $uv \in (E^{\sigma_1} \cup T^{\sigma_1}) \cup \cdots \cup (E^{\sigma_n} \cup T^{\sigma_n})$. We also know that $uv$ is a generator in $\mathbb{G}^{\GB, \mathcal{T}}$. So we define $\eta^{-1}(uv)=uv$. Similarly, this is a well-defined map and is also a homomorphism. The fact that both $\eta \circ \eta^{-1}$ and $\eta^{-1} \circ \eta$ are identity maps can be easily checked. Hence $\eta$ is an isomorphism, as required.
\end{proof}

We now introduce some additional conditions on PCSP templates. If for a template $(\bA,\bB)$ there is a pair of permutations $\Zmap^\bA,\Zmap^\bB$ on $R^\bA$, $R^\bB$ respectively, satisfying conditions (M1)--(M3) below, $(\bA,\bB)$ is said to be a \emph{$\Zmap$-template}.
\begin{itemize}
\item[(M1)]
$\Zmap^\bA,\Zmap^\bB$ preserve faces, that is, for any faces $\sg,\sg'$ of $\cH^\bA,\cH^\bB$, respectively, $\Zmap^\bA(\sg),\Zmap^\bB(\sg')$ are also faces of $\cH^\bA,\cH^\bB$, respectively.
\item[(M2)]
$\Zmap^\bA,\Zmap^\bB$ permute with polymorphisms, that is, for any $n$-ary $f\in \Pol(\bA,\bB)$ it holds that 
$
f(\Zmap^\bA(t_1)\zd\Zmap^\bA(t_n))=\Zmap^\bB(f(\vc tn))
$
for any $\vc tn\in R^\bA$.
\item[(M3)]
There is $\ell>1$, called the order of the pair $\mu^\bA,\mu^\bB$, such that for any faces $\sigma$, $\sigma^\prime$ of $\cH^\bA,\cH^\bB$, respectively, $(\Zmap^\bA)^{\ell}(\sg)=\sg,(\Zmap^\bB)^{\ell}(\sg')=\sg'$, and $(\Zmap^\bA)^i(\sg)\ne\sg,(\Zmap^\bB)^i(\sg')\ne\sg'$ for $i\in[\ell-1]$.

\end{itemize}

We may write $\Zmap$ instead of $\Zmap^{\bA}$ or $\Zmap^{\bB}$ when it is clear from the context to which one we are referring.
Observe that $\ell$ is the order of $\Zmap^\bA$ as a permutation because every vertex is a face and so (M3) applies. Also, $\Zmap$ cannot be the identity since $\ell >1$ by (M3). We can always assume that $\ell$ in (M3) is a prime number. Indeed, if not, there is a prime number $p$ dividing $\ell$. For any faces $\sigma \in \cH^\bA, \sigma' \in \cH^\bB$, we can now define $\Zmap_1^\bA(\sigma):= (\Zmap^\bA)^{\frac{\ell}{p}}(\sigma)$, $\Zmap_1^\bB(\sigma'):= (\Zmap^\bB)^{\frac{\ell}{p}}(\sigma')$ whose order is $p$. We sometimes need only one structure from the $\Zmap$-template $(\bA,\bB)$; in this case we call it a $\Zmap$-structure. The simplicial complex $\GB$ of the $\Zmap$-structure $\bB$ is said to be \textit{$\Zmap$-connected} if for any vertex $t$ in $R^\bB$ there is a path from $t$ to $\Zmap(t)$ in $\GGB$.


Let $\mathcal{H}^\bA$ be $\Zmap$-connected, and let us fix a vertex $t$ in $R^\bA$. For any $t_1\in R^\bA$ that has a path to $t$, we construct a cycle originating at $t$ as follows. Let $P= t, \ldots, t_1$ be a path from $t$ to $t_1$ and let $P^{-1}$ be the reverse of $P$. Next, choose a path $t_1, t_2 \ldots, t_k=\Zmap(t_1)$ from $t_1$ to $\Zmap(t_1)$ that exists, because $\mathcal{H}^\bA$ is $\Zmap$-connected. Then consider the cycle
$$
C=P, t_1, \ldots, t_{k-1}, \Zmap(t_1), \ldots, \Zmap(t_{k-1}), \ldots, \Zmap^{\ell-1}(t_1), \ldots, \Zmap^{\ell-1}(t_{k-1}), t_1, P^{-1}.
$$
A cycle of this form is called a \textit{$\Zmap$-cycle starting at $t$}. 

\begin{lemma} \label{lem: Zmap-cycles not null-homotopic}
Let $\bB$ be a $\Zmap$-structure as above, $\GGB$ connected, and let any two maximal faces of $\GB$ have at most one common vertex. Then any $\Zmap$-cycle in $\GGB$ is not null-homotopic\footnote{Also, see \cite{Wrochna20:homomorphism} for related work.}.
\end{lemma}

\begin{proof}
As was observed, without loss of generality, $\ell$, the order of $\Zmap$, can be assumed to be a prime number.
Fix $t\in\rel^\bB$ and let
\begin{align*}
C^\prime &=P, u_1, u_2, \ldots, u_{k-1}, \Zmap(u_1), \Zmap(u_2), \ldots, \Zmap(u_{k-1}), \Zmap^2(u_1), \ldots, \Zmap^{\ell-1}(u_1), \\
& \hspace*{10mm}\Zmap^{\ell-1}(u_2), \ldots, \Zmap^{\ell-1}(u_{k-1}), u_1, P^{-1}
\end{align*}
be a $\Zmap$-cycle starting at $t$ where $P$ is a path from $t$ to $u_1$ and $P^{-1}$ is the reverse of $P$.
Also, let
\begin{align*}
C &=u_1, u_2, \ldots, u_{k-1}, \Zmap(u_1), \Zmap(u_2), \ldots, \Zmap(u_{k-1}), \Zmap^2(u_1), \ldots, \Zmap^{\ell-1}(u_1),\\
& \hspace*{10mm}\Zmap^{\ell-1}(u_2), \ldots, \Zmap^{\ell-1}(u_{k-1}), u_1
\end{align*}
be the cycle obtained from $C^\prime$ by removing the paths $P$ and $P^{-1}$.
For an edge $uv \in E(C)$, let $N(uv)$ be the number of times $C$ traverses $uv$ from $u$ to $v$. First of all, we prove that there is an edge $uv \in E(C)$ with $N(uv) \neq N(vu)$.

\begin{claim} \label{claim: N(uv)}
There is an edge $uv \in E(C)$ with $N(uv) \neq N(vu)$.
\end{claim}
\begin{claimproof}
For any $i \in \{0, \ldots \ell -1 \}$, consider the following path in $C$ which is from $\Zmap^i(u_1)$ to $\Zmap^{i+1}(u_1)$:
$$S_i=\Zmap^i(u_1), \Zmap^i(u_2) \ldots, \Zmap^i(u_{k-1}), \Zmap^{i+1}(u_1).$$
Let $X_0 = \{v_0, \ldots, v_n\}$ be a set of a maximum number of vertices in $S_0$ such that $\Zmap^r(v_i) \neq v_j$ for any $0 \leq i,j \leq n$ and $r \in [\ell - 1]$.
Next, let $X_i = \{\Zmap^i(v_0), \ldots, \Zmap^i(v_n)\}$ for $i \in [\ell-1]$. Clearly, $X_i \cap X_j = \emptyset$ for any $i \neq j$ and $X_0 \cup \cdots \cup X_{\ell-1}= V(C)$, because $X_0$ is a maximum set. For each edge $uv \in E(C)$, there exist $i,j \in \{0,\ldots,\ell-1\}$ such that $u \in X_i, v \in X_j$. In such a case, $uv$ is said to be a \textit{$k$-edge} where $k= (j - i) \bmod \ell$ (then $vu$ is a $(\ell-k)$-edge). So if $uv$ is an $k$-edge and $u \in X_i$, then $v \in X_{(i + k) \bmod \ell}$. Let $m_i^j$ be the number of $i$-edges in $S_j$. Since $u_1$ and $\Zmap(u_1)$ do not belong to the same set $X_i$ for some $i \in \{0, \ldots, \ell-1 \}$, we must have 
$$
\sum_{i=0}^{\ell-1} i \times m_i^0 \not\equiv 0\pmod\ell.
$$
Moreover, since $C$ is a $\Zmap$-cycle, the number of $i$-edges in $S_0$ is equal to the number of $i$-edges in $S_i$ for any $i \in [\ell-1]$. Therefore, we have $m_i^0 = m_i^1 = \cdots = m_i^{\ell-1}$ for any $i \in \{ 0,\ldots,\ell-1\}$.
We now have two cases depending on whether $\ell > 2$ or $\ell=2$.
\begin{itemize}
\item 
$\ell > 2$: Recall that $\ell$ is a prime number which implies that it is odd. If $m_i^0=m_{\ell-i}^0$ for all $i \in [\ell-1]$, then we have
$$
\sum_{i=0}^{\ell-1} i \times m_i^0 \equiv \sum_{i=1}^{\ell-1} i \times m_i^0\equiv \sum_{i=1}^{\frac{\ell-1}{2}} i \times m_i^0 + (\ell-i) \times m_{\ell-i}^0 \equiv 0 \pmod \ell .
$$
This is a contradiction because we already know $\sum_{i=0}^{\ell-1} i \times m_i^0 \not\equiv 0\pmod\ell$. Therefore, there must exist $r \in [\ell-1]$ such that $m_r^0 > m_{\ell-r}^0$ which implies $m_r^j > m_{\ell-r}^j$ for any $j \in \{ 0,\ldots,\ell-1 \}$ as well. So the number of $r$-edges in $C$ (which is equal to $m_r^0 + \cdots + m_r^{\ell-1}$) is more than the number of $(\ell-r)$-edges in $C$ (which is equal to $m_{\ell-r}^0 + \cdots + m_{\ell-r}^{\ell-1}$). Hence there must exist an $r$-edge $uv$ with $N(uv) \neq N(vu)$, as required.

\item 
$\ell = 2$: For any $1$-edge $uv \in E(C)$, note that $\Zmap(u) \Zmap(v)$ is also a $1$-edge in $C$. 
For $i \in \{0,1 \}$, suppose $N_i(uv)$ be the number of times $S_i$ traverses $uv$ from $u$ to $v$. Let
\begin{eqnarray*}
M_{\{uv, \Zmap(u) \Zmap(v)\}}:= 
N_0(uv) - N_0(vu) + N_0(\Zmap(u) \Zmap(v)) - N_0(\Zmap(v) \Zmap(u)) =\\
= N_0(uv) - N_0(vu) + N_1(uv) - N_1(vu) = N(uv) - N(vu),
\end{eqnarray*}
where we get the first equality because $N_0(\Zmap(u) \Zmap(v)) = N_1(uv)$ and the second equality because $N_0(uv) + N_1(uv) = N(uv)$.
Since for any edge $uv$ with $u \neq v$, $uv$ and $\Zmap(u) \Zmap(v)$ are different edges, i.e., $\{u,v\} \cap \{ \Zmap(u), \Zmap(v) \} = \emptyset$, we can find a set of $1$-edges $E_1$ such that if $uv \in E_1$, then $\Zmap(u) \Zmap(v) \in E_1$, $vu \notin E_1$ (and $\Zmap(v) \Zmap(u) \notin E_1$ which can be deduced from the previous condition). Moreover, for any $1$-edge $uv$ in the path $S_0$, either $uv \in E_1$ or $vu \in E_1$. Now, it is not hard to see that
$$
\sum_{\{uv, \Zmap(u) \Zmap(v)\} \subseteq E_1} M_{\{uv, \Zmap(u) \Zmap(v)\}} \equiv m_1^0 \pmod 2.
$$
Since $u_1$ and $\Zmap(u_1)$ are not in the same set, we must have $m_1^0\not\equiv0 \pmod 2$. Therefore, there exists $uv \in E_1 \subseteq E(C)$ such that $(N(uv) - N(vu)) \not\equiv 0 \pmod 2$. Hence $N(uv) \neq N(vu)$, as desired.
\end{itemize}
\end{claimproof}

We continue with the proof of Lemma~\ref{lem: Zmap-cycles not null-homotopic}. Recall that $C^\prime$ is a $\Zmap$-cycle in $\GB$ starting at $t$ and $C$ is a cycle obtained from $C^\prime$ by removing the initial and final paths $P,P^{-1}$. Let $\sigma_1, \ldots, \sigma_n$ be the maximal faces in $\GB$. Note that for any edge in $\GB$, we must have $\{u,v\} \subseteq \sigma_i$ for some $i \in [n]$ which implies $E(\GB)= \{ uv \; | \; u,v \in \sigma_i, i \in [n] \}$.
Let $T^{\sigma_i}$ be some spanning tree of $\sigma_i$, $i\in[n]$, with $\Zmap(T^{\sigma_i})=T^{\sigma_j}$ where $\Zmap(\sigma_i)=\sigma_j$. Such spanning trees $T^{\sigma_1}, \ldots, T^{\sigma_n}$ exist since $\GB$ is $\Zmap$-structure and any two maximal faces have at most one common tuple.
For any spanning tree $\mathcal{T}$ and any edge $uv \in E(\GB)$, there is a path from $u$ to $v$ with edges in $E=E(T^{\sigma_1}) \cup \cdots \cup E(T^{\sigma_n})$ (since $u,v \in \sigma_i$ for some $i \in [n]$ and $T^{\sigma_i}$ is a spanning tree containing both $u$ and $v$).
So any edge is homotopic to a path whose edges belong to $E$; thus, any $\Zmap$-cycle is homotopic to a $\Zmap$-cycle whose edges belong to $E$. Therefore, we only need to prove the lemma for $\Zmap$-cycles whose edges belong to $E$.

By Claim~\ref{claim: N(uv)}, there is an edge $xy \in E(C)$ with $N(xy) \neq N(yx)$. Since $N(xy) \neq N(yx)$, $C$ should remain connected if we remove $xy$ from $C$. Let $\mathcal{T}$ be a spanning tree with $E(\mathcal{T}) \subseteq E \setminus \{ xy \}$. By Lemma~\ref{lem: conditions on A to be a free group}, $\mathbb{G}^{\GB, \mathcal{T}}$ is a free group with generators $uv$ for $\{u,v\} \in V(\GB)$, $u \neq v$, and $uv \in E \setminus E(\mathcal{T})$.
Recall that $\rho^{t,\mathcal{T}}: \pi_1(\GB) \rightarrow \mathbb{G}^{\GB, \mathcal{T}}$ is an isomorphism with
$$
\rho^{t,\mathcal{T}}: \pi_1(\GB)([t=x_1,\ldots,x_p, x_1=t]) = (x_1 x_2) \cdot (x_2 x_3) \cdot \ldots \cdot (x_p x_1)
$$
and $[C^\prime] \in \pi_1(\GB)$.
As $xy \in E \setminus E(\mathcal{T})$ and $N(xy) \neq N(yx)$, we obtain $\rho^{t,\mathcal{T}}([C^\prime]) \neq \mathbbm{1}$ which proves that $C^\prime$ is not null-homotopic, as required.
\end{proof}
\begin{corollary}
Let $\bB$ be a $\Zmap$-structure as above, $\GGB$ connected, and with faces of size at most $2$. Then any $\Zmap$-cycle in $\GGB$ is not null-homotopic.
\end{corollary}

\begin{corollary} \label{col: nondegenerate template}
Let $(\bA,\bB)$ be a $\Zmap$-template as above and $\mathcal{H}^\bA$ $\Zmap$-connected. Let $C$ be a $\Zmap$-cycle in $R^\bA$. Suppose $\GGB$ is connected and any two maximal faces in $\GB$ have at most one vertex in common. Then $(\bA,\bB)$ is non-degenerate with respect to $R,C$.
\end{corollary}

\begin{proof}
Fix $t'\in\rel^\bA, t\in\rel^\bB$ and let
\begin{align*}
C'_1 &=P, u_1, u_2, \ldots, u_{k-1}, \Zmap(u_1), \Zmap(u_2), \ldots, \Zmap(u_{k-1}), \Zmap^2(u_1), \ldots, \Zmap^{\ell-1}(u_1),\\
&\hspace*{10mm} \Zmap^{\ell-1}(u_2), \ldots, \Zmap^{\ell-1}(u_{k-1}), u_1, P^{-1}
\end{align*}
be a $\Zmap$-cycle in $\GA$ starting at $t'$ where $P$ is a path from $t'$ to $u_1$ and $P^{-1}$ is the reverse of $P$. Let $C_1$ be the cycle obtained from $C'_1$ by removing $P$ and $P^{-1}$, i.e., 
\begin{align*}
C_1 &=u_1, u_2, \ldots, u_{k-1}, \Zmap(u_1), \Zmap(u_2), \ldots, \Zmap(u_{k-1}), \Zmap^2(u_1), \ldots, \Zmap^{\ell-1}(u_1),\\
&\hspace*{10mm} \Zmap^{\ell-1}(u_2), \ldots, \Zmap^{\ell-1}(u_{k-1}), u_1.
\end{align*}
For any path $L= v_1, \ldots, v_n$ in $\GA$ and any homomorphism $g: \bA \rightarrow \bB$, let $g(L)=g(v_1),\ldots,g(v_n)$ be a path in $\GB$. For any homomorphism $g:\bA \rightarrow \bB$, consider the cycle $C$ in $\GB$ given by
$$
C= Q, g(C'_1), Q^{-1}= Q, g(P), g(C_1), g(P^{-1}), Q^{-1}
$$
where $Q$ is a path from the tuple $t$ to $g(t^\prime)$ and $Q^{-1}$ is the reverse of $Q$. The cycle $C$ is a $\Zmap$-cycle in $\GB$ since $g(\Zmap(u))= \Zmap(g(u))$ for any tuple $u$ in $R^\bA$ by (M2). So $C$ is not-null-homotopic by Lemma~\ref{lem: Zmap-cycles not null-homotopic} and the result follows.
\end{proof}

Combining Theorem~\ref{Main Theorem}, Corollary~\ref{col: nondegenerate template}, and and Lemma~\ref{lem: conditions on A to be a free group} we obtain the following theorem.

\begin{theorem} \label{Main Theorem 2}
Let $(\bA,\bB)$ be a $\Zmap$-template, $\GA$ $\Zmap$-connected, and $\GGB$ connected. If any two maximal faces in $\GB$ have at most one vertex in common, then $\PCSP(\bA,\bB)$ is NP-hard.
\end{theorem}

\section{Applications}\label{sec:applications}

\subsection{New proofs of the existing results}

\textbf{The absence of minion homomorphisms.}
As the first application, we use our techniques to prove the nonexistence of a minion homomorphism.
Some of the known NP-hard proofs of PCSP problems, such as $\PCSP(\textbf{C}_{2k+1},\textbf{K}_3)$~\cite{Kro2020}, $\PCSP(\textbf{C}_{2k+1},\textbf{K}_4)$~\cite{Avvakumov25}, and $\PCSP(\textbf{LO}_3,\textbf{LO}_4)$~\cite{Filakovsky23:hardness} are based on finding a minion homomorphism from the polymorphisms of the PCSP template to $\Pol(\mathbb{Z},\textbf G)$ for some group $\textbf G$. We prove that this approach does not work for AGCP.

The proof of the following proposition is based on the fact that if two polymorphisms are different in exactly one point, then they have the same image with respect to minion homomorphism $\xi$. A similar property of topological maps associated with polymorphisms was also mentioned in~\cite{Filakovsky23:hardness}.

\begin{proposition}\label{no minion homomorphism for the AGCP}
For $k \geq 3, t \geq k^3$, any non-trivial group $\textbf G$, any minion homomorphism $\xi: \Pol(\textbf{K}_k,\textbf{K}_t) \rightarrow \Pol(\mathbb{Z},\textbf G)$, and any $f \in \Pol(\textbf{K}_k,\textbf{K}_t)$, the function $\xi(f)$ is constant.
\end{proposition}

\begin{proof}
By contradiction, let $\xi$ be such a minion homomorphism. For any $f \in \Pol(\textbf{K}_k,\textbf{K}_t)$, let $\sel(f)$ be the essential coordinates of $\xi(f)$. It is obvious that $|\sel(f)| \geq 1$ since $\xi(f)$ is not a constant function. We say two $n$-ary polymorphisms $f,g \in \Pol(\textbf{K}_k,\textbf{K}_t)$ are of the same type if $f(x, \ldots, x) = g(x, \ldots, x)$ for any $x \in \textbf{K}_k$.

\setcounter{theorem}{35}
\begin{claim} \label{changing one point does not change the image}
Let $f,g \in \Pol(\textbf{K}_k,\textbf{K}_t)$ be two polymorphisms of the same type and of arities $n \geq 3$. Suppose $f$ and $g$ are different in exactly one point, i.e., there exists a unique sequence of elements $a_1, \ldots, a_n \in \textbf{K}_k$ such that $f(a_1,\ldots,a_n) \neq g(a_1, \ldots, a_n)$. Then $\xi(f) = \xi(g)$.
\end{claim}
\begin{claimproof}
First of all, note that there exists $p,q \in [n]$ such that $a_p \neq a_q$. For any $i \in [n]$, we set $\xi(f)_i:=\xi(f)(0,\ldots,0,1,0,\ldots,0)$ and $\xi(g)_i:=\xi(g)(0,\ldots,0,1,0,\ldots,0)$ where there is a $1$ in the $i$-th coordinate and $0$'s in other coordinates. To obtain a  contradiction, suppose $\xi(f) \neq \xi(g)$, then there is a set $S$ with $|S| \geq 1$ such that $\xi(f)_i \neq \xi(g)_i$ for any $i \in S$ and $\xi(f)_i = \xi(g)_i$ for $i \in [n] \setminus S$. We have three cases based on the number of elements in $S$.
\begin{itemize}
\item 
$|S|=1$: let $S=\{i\}$. There exists $j \in [n]$ such that $a_i \neq a_j$. Let $\pi: [n] \rightarrow [n-1]$ be a surjective map with $\pi(i)=\pi(j)=1$. Then $\xi(f^\pi)_1 = \xi(f)_i \cdot \xi(f)_j$ and $\xi(g^\pi)_1 = \xi(g)_i \cdot \xi(g)_j$. Since $\xi(f)_i \neq \xi(g)_i$ and $\xi(f)_j=\xi(g)_j$, we get $\xi(f^\pi)_1 \neq \xi(g^\pi)_1$ which implies $\xi(f^\pi) \neq \xi(g^\pi)$. This is a contradiction since $f^\pi=g^\pi$.

\item 
$|S|=2$: let $S=\{i,j\}$. There exist $l \in [n] \setminus S$ such that either $a_i \neq a_l$ or $a_j \neq a_l$. With no loss of generality assume that $a_i \neq a_l$. Let $\pi: [n] \rightarrow [n-1]$ be a surjective map with $\pi(i)=\pi(l)=1$. Then $\xi(f^\pi)_1 = \xi(f)_i \cdot \xi(f)_l$ and $\xi(g^\pi)_1 = \xi(g)_i \cdot \xi(g)_l$. Since $\xi(f)_i \neq \xi(g)_i$ and $\xi(f)_l=\xi(g)_l$, we get $\xi(f^\pi)_1 \neq \xi(g^\pi)_1$ which implies $\xi(f^\pi) \neq \xi(g^\pi)$. This is a contradiction since $f^\pi=g^\pi$.

\item 
$|S| \geq 3$: let $\{i,j,l\} \subseteq S$. There exists $r \in [n]$ such that $a_i \neq a_r$. If $j\ne l$, let $\pi: [n] \rightarrow [n-1]$ be a surjective map with $\pi(i)=\pi(r)=1$ and $\pi(l)=2$. Then $\xi(f^\pi)_2 = \xi(f)_l$ and $\xi(g^\pi)_2 = \xi(g)_l$. Since $\xi(f)_l \neq \xi(g)_l$, we get $\xi(f^\pi)_2 \neq \xi(g^\pi)_2$ which implies $\xi(f^\pi) \neq \xi(g^\pi)$. This is a contradiction since $f^\pi=g^\pi$. If $r\ne j$ the argument is similar.
\end{itemize}
\end{claimproof}

\begin{claim} \label{one connected component}
Any two ternary polymorphisms $f,g \in \Pol(\textbf{K}_k,\textbf{K}_t)$ of the same type have the same image under $\xi$, i.e., $\xi(f)=\xi(g)$.
\end{claim}
\begin{claimproof}
Let us start with some definitions. Any sequence $a_1,a_2,a_3 \in \textbf{K}_k$ is said to be a point; it is called a constant point if $a_1=a_2=a_3$.
The value of $f(a_1,a_2,a_3)$ is called a color of a point $a_1,a_2,a_3$ which can be from $1$ to $t$.
At each step, we may change the color of $f(a_1,a_2,a_3)$ where $a_1,a_2,a_3$ is a non-constant point such that the resulted coloring of points remains a polymorphism. Since $t \geq k^3$, there is a sequence of actions that changes the colors of one non-constant point at a time such that in the end, any two non-constant points have different colors. With some work, it is not hard to show that there is a sequence of actions that changes the color of one point at a time such that the end polymorphism would be $g$. By Claim~\ref{changing one point does not change the image}, after changing the color of one point, the image of the resulted polymorphism under $\xi$ does not change. Therefore, the images of all the polymorphisms involved are the same, in particular $f$ and $g$, which completes the proof of the claim.
\end{claimproof}

We continue with the proof of Proposition~\ref{no minion homomorphism for the AGCP}. Let $f \in \Pol(\textbf{K}_k,\textbf{K}_t)$ be a polymorphism of arity $n$. Based on the size of $\sel(f)$ we have two cases.
\begin{itemize}
\item 
$|\sel(f)|=1$: let $\sel(f)=\{i\}$ for $i \in [n]$. Let $\pi,\pi':[n] \rightarrow [3]$ be maps with $\pi(i)=1$ and $\pi'(i)=2$. The only essential coordinates of $\xi(f)^\pi$ and $\xi(f)^{\pi'}$ are $1$ and $2$, respectively. Therefore, $\xi(f^\pi) \neq \xi(f^{\pi'})$. Both $f^\pi$ and $f^{\pi'}$ are ternary polymorphisms of the same type. So by Claim~\ref{one connected component}, $\xi(f^\pi) = \xi(f^{\pi'})$, a contradiction.

\item 
$|\sel(f)|\geq 2$: let $\{i,j\} \subseteq \sel(f)$. Let $\pi,\pi':[n] \rightarrow [3]$ be maps with $\pi(i)=\pi(j)=1$, $\pi(l) \neq 1$ for $l \neq i,j$, and $\pi'(i)=1$, $\pi'(l) \neq 1$ for $l \neq i$. Obviously, $\xi(f^\pi)_1=\xi(f)_i \cdot \xi(f)_j$ and $\xi(f^{\pi'})_1=\xi(f)_i$. So $\xi(f^\pi)_1 \neq \xi(f^{\pi'})_1$.
Therefore, $\xi(f^\pi) \neq \xi(f^{\pi'})$. Both $f^\pi$ and $f^{\pi'}$ are ternary polymorphisms of the same type. So by Claim~\ref{one connected component}, $\xi(f^\pi) = \xi(f^{\pi'})$, a contradiction.
\end{itemize}
\end{proof}

\textbf{Constructing $\Zmap$-templates.}
We start with an example of how a $\Zmap$-template can be constructed. A relation $R$ is said to be \emph{cyclic} if for any $(\vc ar)\in R$ the tuple $(a_r,\vc a{r-1})$ also belongs to $R$.
Let $(\bA,\bB)$ be a PCSP template and a relation $R$ such that $R^\bA,R^\bB$ are cyclic is pp-definable in $(\bA,\bB)$. Let a pair $\cH^\bA,\cH^\bB$ of simplicial complexes on $\rel^\bA,\rel^\bB$ be defined as follows: The faces of $\cH^\bA$ ($\cH^\bB$) are the sets of tuples of the form $A_1\tm\dots\tm A_r\sse\rel^\bA$ (respectively, $B_1\tm\dots\tm B_r\sse\rel^\bB$), which is the standard way of defining simplicial complexes in the literature and also in Lemma~\ref{lem:tolerance}. 
For tuples $t=(\vc ar) \in R^\bA, t'=(\vc{a'}r) \in R^\bB$, we define the action of $\Zmap^\bA$ and $\Zmap^\bB$ as follows:
$
\Zmap^\bA(t):=(a_r,a_1\zd a_{r-1})\quad\text{and}\quad 
\Zmap^\bB(t^\prime):=(a'_r,a'_1\zd a'_{r-1}).
$
A tuple $(a_1,\ldots, a_r)$ is called \emph{periodic} if for some $k \in [r-1]$, we have $a_i = a_{(i + k) \bmod r}$ for any $i \in [r]$. (Note that $a_0 = a_r$). For example, any constant tuple is periodic, and if $r$ is a prime number, then only constant tuples are periodic. Observe that if $R^\bA$ has no periodic tuples, then for any face $\sigma$ and $i \in [r-1]$, we have $\Zmap^i(\sigma) \neq \sigma$. (If $\Zmap^i(\sigma) = \sigma$ for some $1 < i < r$, then by Lemma~\ref{lem:tolerance} there exists a tuple $t \in \sigma$ such that $\Zmap^i(t) = t$). Moreover, $\Zmap^r(\sigma) = \sigma$ for any face $\sigma$. 
The same property holds for $R^\bB$, and $\Zmap$ satisfies condition (M3). Conditions (M1) and (M2) are straightforward.
$\Zmap$-template of this type have been intensively studied in the literature including that on the PCSP (see, e.g., \cite{Kro2020}). The majority of the results fit in this framework for $r=2$ and \cite{Filakovsky23:hardness} considers such a construction for $r=3$.

\textbf{The hardness of $\PCSP(\textbf C_{2k+1}, \textbf K_3)$.}
With the observation above (and assuming that $\rel$ is the edge relation of the corresponding graph) the NP-hardness of $\PCSP(\textbf C_{2k+1}, \textbf K_3)$,  \cite{Kro2019,Kro2020} follows from Theorem~\ref{Main Theorem 2}. Indeed, let $\cH^{\mathbf{C}_{2k+1}}$ and $\cH^{\mathbf{K}_3}$ be defined as after Lemma~\ref{lem:tol-classes}. Then $\cH^{\mathbf{C}_{2k+1}}$ and $\cH^{\mathbf{K}_3}$ are both connected, any face in $\cH^{\textbf K_3}$ is of size at most $2$ and there is no periodic tuple in both $R^{\textbf K_3}$ and $R^{\textbf C_{2k+1}}$.

\begin{corollary} {\rm \cite{Kro2019,Kro2020}} \label{PCSP(C_{2k+1}, K_3)}
For $k \geq 1$, $\PCSP(\textbf C_{2k+1}, \textbf K_3)$ is NP-hard.
\end{corollary}

If a graph $\textbf{G}$ does not have $\textbf K_{2,2}$ as a subgraph, then any maximal face in $\cH^{\textbf{G}}$ is of the form $\textbf K_{1,r}$ for some $r \geq 1$. So any two maximal faces of $\cH^\textbf G$ have at most one tuple in common. Therefore we get the following corollary which is a generalization of Corollary~\ref{PCSP(C_{2k+1}, K_3)}.

\begin{corollary} {\rm \cite{Kro2020}}
Let $\textbf G$ be a non-bipartite graph such that there is a homomorphism from $\textbf C_{2k+1}$ to $\textbf G$ that does not have $\textbf K_{2,2}$ as a subgraph. Then $\PCSP(\textbf C_{2k+1}, \textbf G)$ is NP-hard for $k \geq 1$.
\end{corollary}

\smallskip

\textbf{The hardness of $\PCSP(\textbf H_2, \textbf E)$.}
Another application of Theorem~\ref{Main Theorem 2} with ternary relations is defined as follows.
Let $\textbf H_i=\{[i], R^{\textbf H_i}\}$ be the relational structure where $R^{\textbf H_i}$ is a ternary relation that contains all tuples except for the constant tuples. The NP-hardness of $\PCSP(\textbf H_2, \textbf H_k)$ was proved in \cite{Din2005} for any $k \geq 2$.
Let $\textbf E= (\{1,2,3\}; R^{\textbf E})$ be a relational structure with one ternary relation where
$$
R^{\textbf E} = R^{\textbf H_2} \cup \{(a,b,c)\mid |\{a,b,c\}|=3\}.
$$
The six maximal faces of $\cH^\textbf E$ are obtained from $\{(1,2,1),(1,2,2),(1,2,3)\}$ by permutations of the coordinate positions.
Then any two maximal faces in $\cH^{\textbf E}$ have at most $1$ common tuple. The structures $\textbf H_2$ and $\textbf E$ are both connected and $\Zmap$-connected. There is no periodic tuple in both $R^{\textbf H_2}$ and $R^{\textbf E}$. 

\begin{corollary} {\rm \cite{Din2005}}
$\PCSP(\textbf H_2, \textbf E)$ is NP-hard.
\end{corollary}

\subsection{New hardness results through Theorem~\ref{Main Theorem 2}}
In this subsection, we prove several new hardness results that can be obtained from Theorem~\ref{Main Theorem 2}. Note that for the following results, we use the same mapping $\mu$ as was defined after Proposition~\ref{no minion homomorphism for the AGCP}.

\smallskip

\textbf{The hardness of $\PCSP(\bA, \bA_3)$.}
Let $\bA_n=(\{0,1,\ldots,n-1\};R^{\textbf A_n})$ be a relational structure for $n \geq 2$ where $R^{\bA_n}$ is a ternary relation containing all tuples $(a,b,c)$ such that the set $\{a,b,c\}$ has a unique maximum element, see also \cite{Nak2023}. Also, let $\bA=(\{0,1,2\};R^\bA)$ be a relational structure where $R^\bA$ is a ternary relation given by 
\[
R^\bA=R^{\bA_3} \setminus \{(1,1,2),(1,2,1),(2,1,1),(1,0,2),(0,2,1),(2,1,0)\}.
\]
The maximal faces of $\cH^{\bA_3}$ can be found in Fig.~\ref{fig:complexes}(a) (the edges and shaded rhombuses).

As is easily seen, any two faces in $\cH^{\bA_3}$ have at most $1$ common tuple. It is also true for $\bA$ since the set of tuples of $R^\bA$ is a subset of that of $R^{\bA_3}$. The structures $\bA$ and $\bA_3$ are both connected. There is no periodic tuple in both $R^\bA$ and $R^{\bA_3}$. By Theorem~\ref{Main Theorem 2} we obtain the following result.

\begin{corollary}
$\PCSP(\bA, \bA_3)$ is NP-hard.
\end{corollary}

Note that $\PCSP(\bA_2, \bA_3)$ was conjectured to be NP-hard in~\cite{barto2020symmetric}; the conjecture was confirmed in \cite{Kro2025}.

\begin{remark}
Let us define a relational structure $\bA'=(\{0,1,2\}; R^{\bA'})$ where $R^{\bA'}$ is the maximum subset of $R^\bA$ with only symmetric tuples, i.e.,
$$
R^{\bA'} = R^\bA \setminus \{(0,1,2),(2,0,1),(1,2,0) \} = \{ (0,0,1),(0,1,0),(1,0,0),(0,0,2),(0,2,0),(2,0,0) \}.
$$
Clearly, $\bA'$ and $\bA_2$ are homomorphically equivalent. So $\PCSP(\bA_2, \bA_3)$ and $\PCSP(\bA', \bA_3)$ are polynomial time reducible to one another. The same technique that worked for $\PCSP(\bA, \bA_3)$ does not work for $\PCSP(\bA_2, \bA_3)$ nor $\PCSP(\bA', \bA_3)$ since $\bA_2$ and $\bA'$ are neither connected nor $\Zmap$-connected.
\end{remark}

\textbf{The hardness of $\PCSP(\textbf B_n, \textbf D_n)$.}
Before introducing another application, let us define a cyclic tuple. A tuple $t=(t_1, \ldots, t_n)$ on $\{0,1\}$ is said to be a \textit{cyclic tuple} if there exist $i,j$ with $1 \leq i \leq j \leq n$ and $j-i < n-1$ such that one of the following conditions hold:
\begin{itemize}
\item 
Either $t_k = 1$ for any $i \leq k \leq j$ and $t_k = 0$ otherwise, or
\item 
$t_k = 0$ for any $i \leq k \leq j$ and $t_k = 1$ otherwise.
\end{itemize}
In the first case, $t$ is said to be the $(i,j)$-cyclic tuple, and in the second case the $(j,i)$-cyclic tuple. If the number of $1$'s in $t$ is equal to $r$, then $t$ is called an $r$-cyclic tuple.

Now, let $\textbf B_n=(\{0,1\}; R^{\textbf B_n})$ and $\textbf D_n=(\{0,1\}; R^{\textbf D_n})$ be relational structures with one relations of arity $n \geq 3$ where $R^{\textbf B_n}$ only contains all $1$-cyclic and $2$-cyclic tuples, and $R^{\textbf D_n}$ contains all cyclic tuples. (Note that we do not have any $0$-cyclic or $n$-cyclic tuples according to the definition.) Fig.~\ref{fig:complexes}(b) shows the simplicial complex of the structure $\mathbf D_4$.

\begin{corollary}
$\PCSP(\textbf B_n, \textbf D_n)$ is NP-hard.
\end{corollary}

\begin{proof}
The structures $\textbf B_n$ and $\textbf D_n$ are both connected (and therefore $\Zmap$-connected), and there is no periodic tuple in both $R^{\textbf B_n}$ and $R^{\textbf D_n}$. To prove the NP-hardness of $\PCSP(\textbf B_n, \textbf D_n)$ for $n \geq 3$, 
it remains to show that any two maximal faces in $\cH^{\textbf D_n}$ have at most $1$ tuple in common.

First, we define an operator $\oplus_n$ on $[n]$ where $a \oplus_n b:= a + b \bmod n$ where $0$ and $n$ are identified. Similarly, we define an operator $\ominus$ on $[n]$ where $a \ominus_n b:= a - b \bmod n$ where $0$ and $n$ are identified. For any $i,j$ with $1 \leq i \leq j \leq n$ and $j - i < n-2$, let $\sigma_{i,j}$ be the face that contains the $(i,j)$-cyclic, $(i,j \oplus_n 1)$-cyclic, $(i \ominus_n 1, j)$-cyclic, and $(i \ominus_n 1, j \oplus_n 1)$-cyclic tuples. Suppose $\sigma_{j,i}$ be the face where all $0$'s and $1$'s in all tuples of $\sigma_{i,j}$ are swapped. It is not hard to see that $\sigma_{i,j}$ and $\sigma_{j,i}$ for $1 \leq i \leq j \leq n$ and $j-i < n-2$ are all maximal faces of $\cH^{\textbf D_n}$ and any two maximal faces in $\cH^{\textbf{D}_n}$ have at most $1$ common tuple. Hence
$\PCSP(\textbf B_n, \textbf D_n)$ is NP-hard for any $n \geq 3$.
\end{proof}


\begin{figure}[t]
\centerline{\includegraphics[height=5cm, keepaspectratio]{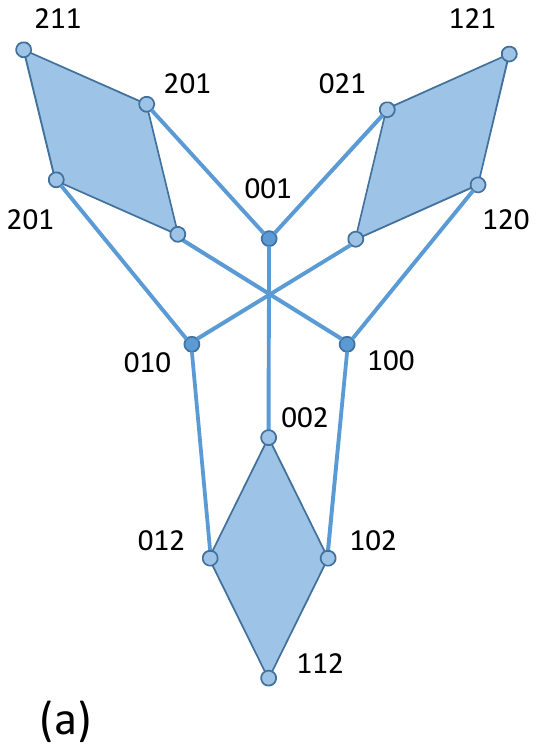}\hspace{-2cm}\includegraphics[height=5cm, keepaspectratio]{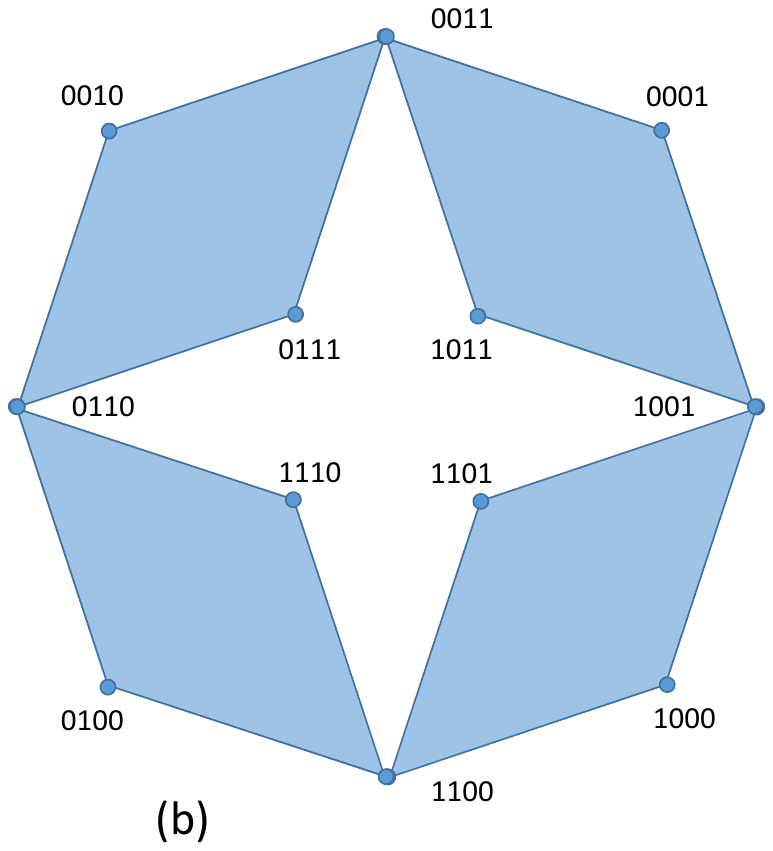}}
\vspace*{-1cm}
\caption{}\label{fig:complexes}
\end{figure}

\smallskip


\bibliographystyle{plainurl}
\bibliography{refs.bib}

\end{document}